\documentclass[11pt]{article}

\usepackage{ifthen}
\usepackage{mdwlist}
\usepackage{amsmath,amssymb,amsfonts,amsthm}
\usepackage{bm}
\usepackage{mathtools}
\usepackage{url}
\usepackage[hidelinks]{hyperref}
\usepackage{enumitem}
\usepackage{graphicx}
\usepackage{xspace}
\usepackage{verbatim}
\usepackage{color}
\usepackage{latexsym}
\usepackage{epsfig}
\usepackage{bm}
\usepackage[noend]{algpseudocode}
\usepackage{algorithm}
\usepackage{etoolbox}
\usepackage{tikz}
\usepackage{pgfplots}
\pgfplotsset{compat=1.7}
\usepackage{hyperref}
\usepackage{subcaption}

\newcommand{\Var}{\mathbf{Var}}
\newcommand{\var}{\mathbf{Var}}
\newcommand{\E}{\mathbf{E}}
\newcommand{\pr}{\mathbf{Pr}}

\def\eps{\epsilon}

\newcommand{\erfc}{\mathrm{erfc}}

\newcommand{\prob}[2][]{\text{\bf Pr}\ifthenelse{\not\equal{}{#1}}{_{#1}}{}\!\left[#2\right]}
\newcommand{\expect}[2][]{\text{\bf E}\ifthenelse{\not\equal{}{#1}}{_{#1}}{}\!\left[#2\right]}

\newcommand{\Tr}{{\mathrm {Tr}}}

\newcommand{\wt}[1]{{\widetilde{#1}}}
\newcommand{\wh}[1]{{\widehat{#1}}}

\newtheorem{theorem}{Theorem}[section]
\newtheorem{lemma}[theorem]{Lemma}

\newtheorem{proposition}[theorem]{Proposition}
\newtheorem{corollary}[theorem]{Corollary}
\newtheorem{claim}[theorem]{Claim}
\newtheorem{definition}[theorem]{Definition}

\newtheorem{fact}[theorem]{Fact}

\newcommand{\ignore}[1]{}

\newcommand{\bg}[1]{\medskip\noindent{\bf #1}}

\definecolor{Red}{rgb}{1,0,0}

\newcommand{\oldbound}[1]{{}}

\newcommand{\normalpdf}{\ensuremath{\mathcal{N}}}

\newcommand{\Cov}{\mbox{Cov}}

\newcommand{\todo}[1]{#1}
\renewcommand{\epsilon}{\varepsilon}
\newcommand{\tr}{\mathrm{tr}}

\DeclareMathOperator{\R}{\mathbb{R}}

\DeclareMathOperator{\Z}{\mathbb{Z}}

\DeclareMathOperator{\poly}{poly}

\renewcommand{\[ }{\begin{eqnarray*}}
\renewcommand{\]}{\end{eqnarray*}}

\newcommand{\Ot}{\widetilde{O}}

\definecolor{darkpastelred}{rgb}{0.76, 0.23, 0.13}

\algnewcommand\INPUT{\item[{\textbf {input:}}]}
\algnewcommand\OUTPUT{\item[{\textbf{output:}}]}

\def\colorful{1}

\ifnum\colorful=0
\newcommand{\new}[1]{{\color{red} #1}}

\else
\newcommand{\new}[1]{{#1}}

\fi

\newcommand{\eqdef}{\stackrel{{\mathrm {\footnotesize def}}}{=}}

\usepackage{ifthen}
\usepackage{mdwlist}

\usepackage{amsmath,amssymb,amsfonts,amsthm}
\usepackage{bm}
\usepackage{enumitem}
\usepackage{graphicx}
\usepackage{xspace}
\usepackage{verbatim}
\usepackage{algorithm}
\usepackage{thm-restate}
\usepackage{latexsym}
\usepackage{epsfig}
\usepackage[symbol]{footmisc}

\newcommand{\tmu}{\tilde{\mu}}
\newtheorem{conjecture}{Conjecture}
\newcommand{\wtSigma}{\widetilde{\Sigma}}
\newcommand{\vvec}{\mathsf{vec}}

\newcommand{\supp}{\mathsf{Supp}}

\title{Outlier-Robust High-Dimensional Sparse Estimation \\ 
via Iterative Filtering}
\author{
Ilias Diakonikolas\thanks{Supported by NSF Award CCF-1652862 (CAREER) and a Sloan Research Fellowship.}\\
University of Wisconsin-Madison\\
{\tt ilias@cs.wisc.edu}\\
\and
Sushrut Karmalkar\thanks{Supported by NSF Award CNS-1414023.}\\
UT Austin\\
{\tt s.sushrut@gmail.com}
\and
Daniel Kane\thanks{Supported by NSF Award CCF-1553288 (CAREER) and a Sloan Research Fellowship}\\
University of California, San Diego\\
{\tt dakane@ucsd.edu}\\
\and
Eric Price\thanks{Supported in part by NSF Award CCF-1751040 (CAREER).}\\ 
UT Austin\\
{\tt ecprice@cs.utexas.edu}
\and
Alistair Stewart\thanks{Part of this work was performed while the author was a postdoctoral researcher at 
USC.}\\
Web3 Foundation\\
{\tt stewart.al@gmail.com}
}

\begin{document}

\maketitle

\begin{abstract}
We study high-dimensional sparse estimation tasks in a robust setting where a constant fraction
of the dataset is adversarially corrupted. Specifically, we focus on the fundamental problems of robust
sparse mean estimation and robust sparse PCA. 
We give the first practically viable robust estimators for these problems. 
In more detail, our algorithms are sample and computationally efficient 
and achieve near-optimal robustness guarantees. 
In contrast to prior provable algorithms which relied on the ellipsoid method, 
our algorithms use spectral techniques to iteratively remove outliers from the dataset. 
Our experimental evaluation on synthetic data shows that our algorithms are scalable and 
significantly outperform a range of previous approaches, nearly matching the best error rate without corruptions.
\end{abstract}

\setcounter{page}{0}

\thispagestyle{empty}

\newpage

\section{Introduction} \label{sec:intro}
\subsection{Background} \label{ssec:background}
The task of leveraging sparsity to extract meaningful information from high-dimensional datasets
is a fundamental problem of significant practical importance, motivated by a range of data analysis applications.
Various formalizations of this general problem have been investigated in statistics and machine learning
for at least the past two decades, see, e.g.,~\cite{Hastie15} for a recent textbook on the topic.
This paper focuses on the {\em unsupervised setting} and in particular on estimating the
parameters of a high-dimensional distribution under sparsity assumptions.
Concretely, we study the problems of {\em sparse mean estimation}
and {\em sparse PCA} under natural data generating models.

The classical setup in statistics is that the data was generated by a probabilistic model
of a given type. This is a simplifying assumption that is only approximately valid,
as real datasets are typically exposed to some source of contamination. The field
of robust statistics~\cite{Huber64,Huber09, HampelEtalBook86} aims to design
estimators that are {\em robust} in the presence of model misspecification.
In recent years, designing computationally efficient robust estimators for
high-dimensional settings has become a pressing challenge in a number of applications.
These include the analysis of biological datasets, where natural outliers are
common~\cite{RP-Gen02, Pas-MG10, Li-Science08}
and can contaminate the downstream statistical analysis,
and {\em data poisoning attacks}~\cite{Barreno2010}, where even a small
fraction of fake data (outliers) can substantially degrade the learned model~\cite{BiggioNL12, SteinhardtKL17}.

This discussion motivates the design of robust estimators
that can tolerate a {\em constant} fraction of adversarially corrupted data.
We will use the following model
of corruptions (see, e.g.,~\cite{DKKLMS16}):

\begin{definition} \label{def:adv}
Given $0< \eps < 1/2$ and a family of distributions $\mathcal{D}$ on $\R^d$,
the \emph{adversary} operates as follows: The algorithm specifies some number of samples $N$, and $N$ samples $X_1, X_2, \ldots, X_N$ are drawn from some (unknown) $D \in \mathcal{D}$.
The adversary is allowed to inspect the samples, removes $\eps N$ of them,
and replaces them with arbitrary points. This set of $N$ points is then given to the algorithm.
We say that a set of samples is {\em $\eps$-corrupted}
if it is generated by the above process.
\end{definition}

Our model of corruptions generalizes several other robustness models,
including Huber's contamination model~\cite{Huber64} and the malicious
PAC model~\cite{Valiant:85, keali93}.

In the context of robust sparse mean estimation,
we are given an $\eps$-corrupted set of samples from an unknown mean Gaussian distribution
$\normalpdf(\mu,I)$, where $\mu \in \R^d$ is assumed to be $k$-sparse,
and the goal is to output a hypothesis vector $\wh{\mu}$ that approximates $\mu$
in $\ell_2$-norm. In the context of robust sparse PCA
(in the spiked covariance model), we are given an $\eps$-corrupted set of samples from
$\normalpdf(\mathbf{0}, \rho vv^T)$, where $v \in \R^d$ is assumed to be $k$-sparse
and the goal is to approximate $v$. In both settings, we would like to design computationally
efficient estimators with sample complexity $\poly(k, \log d, 1/\eps)$, i.e., close to the information theoretic
minimum, that achieve near-optimal error guarantees.

Until recently, even for the simplest high-dimensional parameter estimation settings,
no polynomial time robust learning algorithms 
with dimension-independent error guarantees
were known. Two concurrent works~\cite{DKKLMS16, LaiRV16} made the first progress on this front
for the unsupervised setting. Specifically,~\cite{DKKLMS16, LaiRV16} gave the first 
polynomial time algorithms for robustly learning the mean and covariance of
high-dimensional Gaussians and other models. These works focused on the dense regime and
as a result did not obtain algorithms with sublinear sample complexity in the sparse setting.
Building on~\cite{DKKLMS16}, more recent work~\cite{BDLS17} obtained sample efficient
polynomial time algorithms for the robust {\em sparse} setting, 
and in particular for the problems of robust sparse mean estimation and robust sparse PCA studied in this paper.
These algorithms are based the unknown convex programming methodology of~\cite{DKKLMS16}
and in particular {\em inherently rely on the ellipsoid algorithm}. Moreover, the separation oracle
required for the ellipsoid algorithm turns out to be another convex program --- corresponding
to an SDP to solve sparse PCA. As a consequence, the running time of these algorithms,
while polynomially bounded, is impractically high.

\subsection{Our Results and Techniques} \label{ssec:results}

The main contribution of this paper is the design of significantly faster robust estimators 
for the aforementioned high-dimensional sparse problems.
More specifically, our algorithms are iterative and each iteration involves a simple spectral
operation (computing the largest eigenvalue of an approximate matrix). Our algorithms
achieve the same error guarantee as~\cite{BDLS17} with similar sample complexity.
At the technical level, we enhance the {\em iterative filtering methodology} of~\cite{DKKLMS16} to
the sparse setting, which we believe is of independent interest and could lead to faster algorithms for
other robust sparse estimation tasks as well.

For robust sparse mean estimation, we show:

\begin{theorem}[Robust Sparse Mean Estimation]\label{thm:mean-informal}
Let $D \sim \normalpdf(\mu,I)$ be a Gaussian distribution on $\R^d$ with
unknown $k$-sparse mean vector $\mu$, and $\eps > 0$.
Let $S$ be an $\eps$-corrupted set of samples from $D$ of size
$N=\widetilde{\Omega}(k^2 \log(d)/\eps^2)$. There exists an algorithm that,
on input $S$, $k$, and $\eps$ runs in polynomial time returns $\widehat{\mu}$
such that with probability at least $2/3$ it holds $\|\widehat{\mu}-\mu\|_2 = O(\eps\sqrt{\log(1/\eps)}).$
\end{theorem}

Some comments are in order. First, the sample complexity of our algorithm is asymptotically 
the same as that of~\cite{BDLS17}, and matches the lower bound of~\cite{DKS17-sq} 
against Statistical Query algorithms for this problem. The major advantage of our algorithm 
over~\cite{BDLS17} is that while their algorithm made use of the ellipsoid method, 
ours uses only spectral techniques and is scalable. 

For robust sparse PCA in the spiked covariance model, we show:

\begin{theorem}[Robust Sparse PCA]\label{thm:sparse-PCA-informal}
Let $D \sim \normalpdf(\mathbf{0}, I+ \rho v v^T)$ be a Gaussian distribution on $\R^d$
with spiked covariance for an unknown $k$-sparse unit vector $v$,
and $0 < \rho < O(1)$. For $\eps > 0$, let $S$ be an $\eps$-corrupted set of samples from $D$ of size
$N={\Omega}(k^4 \log^4(d/\eps)/\eps^2)$. There exists an algorithm that,
on input $S$, $k$, and $\eps$, runs in polynomial time and returns $\hat{v} \in \R^d$
such that with probability at least $2/3$ we have that $\|\hat{v} \hat{v}^T - vv^T\|_F = O\left(\eps \log(1/\eps)/\rho\right)$. 
\end{theorem}

The sample complexity upper bound in Theorem~\ref{thm:sparse-PCA-informal} 
is somewhat worse than the information theoretic optimum of $\Theta (k^2 \log d / \eps^2)$. 
While the ellipsoid-based algorithm of~\cite{BDLS17} achieves near-optimal sample complexity 
(within logarithmic factors), our algorithm is practically viable as it only uses spectral operations.
We also note that the sample complexity in our above theorem is not known to be optimal 
for our algorithm. It seems quite plausible, via a tighter analysis, that our algorithm in fact 
has near-optimal sample complexity as well.

For both of our algorithms, in the most interesting regime of $k \ll
\sqrt{d}$, the running time per iteration is dominated by the $O(N d^2)$
computation of the empirical covariance matrix.  The number of
iterations is at most $\eps N$, although it typically is much smaller,
so both algorithms take at most $O(\eps N^2 d^2)$ time.

\subsection{Related Work} \label{ssec:related}
There is extensive literature on exploiting sparsity in statistical estimation (see, e.g.,~\cite{Hastie15}).
In this section, we summarize the related work that is directly related to the results
of this paper. Sparse mean estimation is arguably one of the most
fundamental sparse estimation tasks and is closely related to the Gaussian sequence model~\cite{Tsybakov08, J17}.
The task of sparse PCA in the spiked covariance model, initiated in~\cite{J01}, has been extensively
investigated (see Chapter~8 of~\cite{Hastie15} and references therein). In this work, we design
algorithms for the aforementioned problems that are robust to a constant fraction of outliers.

Learning in the presence of outliers is an important goal in
statistics studied since the 1960s~\cite{Huber64}. See, e.g.,~\cite{Huber09, HampelEtalBook86}
for book-length introductions in robust statistics.
Until recently, all known computationally efficient high-dimensional estimators
could tolerate a negligible fraction of outliers, even for the task of mean estimation.
Recent work~\cite{DKKLMS16, LaiRV16} gave the first efficient robust
estimators for basic high-dimensional unsupervised tasks, including
mean and covariance estimation. Since the dissemination of~\cite{DKKLMS16, LaiRV16},
there has been a flurry of research activity on computationally efficient robust learning in high dimensions~\cite{BDLS17, CSV17, DKK+17, DKS17-sq, DiakonikolasKKLMS18, SteinhardtCV18, DiakonikolasKS18-mixtures, DiakonikolasKS18-nasty, HopkinsL18, KothariSS18, PrasadSBR2018, DiakonikolasKKLSS2018sever, KlivansKM18, DKS19-lr, LSLC18-sparse, ChengDKS18, ChengDR18, CDGW19}.

In the context of robust sparse estimation,~\cite{BDLS17} obtained sample-efficient and polynomial time
algorithms for robust sparse mean estimation and robust sparse PCA. The main difference between~\cite{BDLS17}
and the results of this paper is that the~\cite{BDLS17} algorithms use the ellipsoid method (whose separation oracle
is an SDP). Hence, these algorithms are prohibitively slow for practical applications.
More recent work~\cite{LSLC18} gave an iterative method for robust sparse mean estimation, which however
requires multiple solutions to a convex relaxation for sparse PCA in each iteration.
Finally,~\cite{LLC19} proposed an algorithm for robust sparse mean estimation via iterative trimmed hard thresholding.
While this algorithm seems practically viable in terms of runtime,
it can only tolerate $1/(\sqrt{k} \log(n d))$ -- i.e., {\em sub-constant} -- fraction
of corruptions.

\subsection{Paper Organization} 
In Section~\ref{sec:alg}, we describe our algorithms and provide a detailed sketch of their analysis. 
In Section~\ref{sec:experiments}, we report detailed experiments demonstrating the performance 
of our algorithms on synthetic data in various parameter regimes. Due to space limitations, 
the full proofs of correctness for our algorithms can be found in the full version of this paper.

\section{Algorithms} \label{sec:alg}


In this section, we describe our algorithms in tandem with a detailed
outline of the intuition behind them and a sketch of their analysis.
Due to space limitations, the proof of correctness is deferred to the full version of our paper. 

At a high-level, our algorithms use the iterative filtering methodology
of~\cite{DKKLMS16}. The main idea is to iteratively remove a small subset
of the dataset, so that eventually we have removed all the important outliers
and the standard estimator (i.e., the estimator we would have used in the noiseless
case) works. Before we explain our new ideas that enhance the filtering methodology
to the sparse setting, we provide a brief technical description of the approach.

\vspace{-0.2cm}

\paragraph{Overview of Iterative Filtering.}
The basic idea of iterative filtering~\cite{DKKLMS16} is the following:  
In a given iteration, carefully pick some test statistic (such as $v \cdot x$ for a well-chosen $v$). 
If there were no outliers, this statistic would follow a nice distribution (with good concentration properties).  
This allows us to do some sort of statistical hypothesis testing of the
``null hypothesis'' that each $x_i$ is an inlier, rejecting it (and
believing that $x_i$ is an outlier) if $v \cdot x_i$ is far from the
expected distribution.  Because there are a large number of such
hypotheses, one uses a procedure reminiscent of 
the Benjamini-Hochberg procedure~\cite{benjamini1995controlling} 
to find a candidate set of outliers with low \emph{false discovery rate} (FDR), 
i.e., a set with more outliers than inliers in expectation. This procedure looks for a
threshold $T$ such that the fraction of points with test statistic
above $T$ is at least a constant factor more than it ``should'' be.
If such a threshold is found, those points are mostly outliers and
can be safely removed.  The key goal is to judiciously 
design a test statistic such that either the outliers aren't
particularly important---so the naive empirical solution is
adequate---or at least one point will be filtered out.

In other words, the goal is to find a test statistic such that, if the
distribution of the test statistic is ``close'' to what it would be in
the outlier-free world, then the outliers cannot perturb the answer
too much.  An additional complication is that the test statistics
depend on the data (such as $v \cdot x$, where $v$ is the principal
component of the data) making the distribution on inliers also
nontrivial. This consideration drives the sample complexity of the
algorithms.  

In the algorithms we describe below, we use a specific parameterized notion of a good set. We define these precisely in the supplementary material, briefly, any large enough sample drawn from the uncorrupted distribution will satisfy the structural properties required for the set to be good.

We now describe how to design such test statistics for our two sparse settings. 

\paragraph{Notation} 
Before we describe our algorithms, we set up some notation. We define $h_k: \mathbb{R}^d \rightarrow \mathbb{R}^d$ to be the thresholding operator that keeps the $k$ entries of $v$ with the largest magnitude and sets the rest to 0. For a finite set $S$, we will use $a \in_u S$ to mean that $a$ is chosen uniformly at random from $S$. For $M \in \mathbb{R}^d \times \mathbb{R}^d$ and $U \subseteq [d]$, let $M_U$ denote the matrix $M$ restricted to the $U \times U$ submatrix.

\paragraph{Robust Sparse Mean Estimation.}
Here we briefly describe the motivation and analysis of Algorithm~\ref{alg:rsm}, 
describing a single iteration of our filter for the robust sparse mean setting. 

In order to estimate the $k$-sparse mean $\mu$, it suffices to ensure 
that our estimate $\mu'$ has $|v \cdot (\mu'-\mu)|$ small for any $2k$-sparse unit vector $v$. 
The now-standard idea in robust statistics~\cite{DKKLMS16} is that if a small number of corrupted 
samples suffice to cause a large change in our estimate of $v \cdot \mu$, then this must lead 
to a substantial increase in the sample variance of $v \cdot x$, which we can detect.

Thus, a very basic form of a robust algorithm might be to compute a sample covariance matrix
$\wt{\Sigma}$, and let $v$ be the $2k$-sparse unit vector that maximizes
$v^T \wt{\Sigma} v$. If this number is close to $1$, it certifies that our
estimate $\mu'$ --- obtained by truncating the sample mean to its $k$-largest
entries --- is a good estimate of the true mean $\mu$. If not, this will
allow us to filter our sample set by throwing away the values where $v \cdot x$ is furthest from the true mean. 
This procedure guarantees that we have removed more corrupted samples than
uncorrupted ones. We then repeat the filter until the empirical variance in every sparse direction
is close to $1$. 

Unfortunately, the optimization problem of finding the optimal $v$ is
computationally challenging, requiring a convex program. 
To circumvent the need for a convex program, we notice that 
$v^T \wt{\Sigma} v - 1 = (\wt{\Sigma}-I) \cdot (vv^T)$ is
large only if $\wt{\Sigma}-I$ has large entries on the $(2k)^2$ non-zero entries
of $vv^T$. Thus, if the $4k^2$ largest entries of $\wt{\Sigma}-I$ had small 
$\ell_2$-norm, this would certify that no such bad $v$ existed and would allow us
to return the truncated sample mean. In case these entries have large $\ell_2$-norm, 
we show that we can produce a filter that removes more bad samples than good ones. 
Let $A$ be the matrix consisting of the large entries of $\wt{\Sigma}$ (for the moment assume that they are all off
diagonal, but this is not needed). We know that the sample mean
of $p(x) = (x-\mu')^T A (x-\mu') = \wt{\Sigma} \cdot A = \|A\|_F^2$. On the other hand, if
$\mu'$  approximates $\mu$ on the $O(k^2)$ entries in question, we would
have that $\|p\|_2 = \|A\|_F$. This means that if $\|A\|_F$ is reasonably large,
an $\eps$-fraction of corrupted points changed the mean of $p$ from $0$ to
$\|A\|_F^2 = \|A\|_F \|p\|_2$. This means that many of these errors
must have had $|p(x)| \ll \|A\|_F/\eps \|p\|_2$. This becomes very unlikely for good samples if
$\|A\|_F$ is much larger than $\eps$ (by standard results on the
concentration of Gaussian polynomials). Thus, if $\mu'$ is approximately
$\mu$ on these $O(k^2)$ coordinates, we can produce a filter. To ensure this, 
we can use existing filter-based algorithms to
approximate the mean on these $O(k^2)$ coordinates. This results in Algorithm~\ref{alg:rsm}. 
For the analysis, we note that if the entries of $A$ are small, then
$v^T(\wt{\Sigma}-I) v$ must be small for any unit  $k$-sparse $v$, which
certifies that the truncated sample mean is good. 
Otherwise, we can filter the samples using the first kind of filter. 
This ensures that our mean estimate is sufficiently close 
to the true mean that we can then filter using the second kind of filter.




It is not hard to show that the above works if we are given
sufficiently many samples, but to obtain a tight analysis
of the sample complexity, we need a number of subtle technical ideas. 
The detailed analysis of the sample complexity is deferred to the full version of our paper. 

\paragraph{Robust Sparse PCA}
Here we briefly describe the motivation and analysis of Algorithm~\ref{alg:rspca}, 
describing a single iteration of our filter for the sparse PCA setting. 

Note that estimating the $k$-sparse vector $v$ is equivalent to estimating
$\E[XX^T - I] = vv^T$. In fact, estimating $\E[XX^T-I]$ to error $\eps$ 
in Frobenius norm allows one to estimate $v$ within
error $\eps$ in $\ell_2$-norm. Thus, we focus on he task of robustly approximating
the mean of $Y = XX^T - I$.

Our algorithm is going to take advantage of one fact about $X$ that even
errors cannot hide: that $\Var[v \cdot X]$ is large. This is because removing uncorrupted
samples cannot reduce the variance by much more than an $\eps$-fraction,
and adding samples can only increase it. This means that an adversary
attempting to fool our algorithm can only do so by creating other
directions where the variance is large, or simply by adding other
large entries to the sample covariance matrix in order to make it hard
to find this particular $k$-sparse eigenvector. In either case, the
adversary is creating large entries in the empirical mean of $Y$ that
should not be there. This suggests that the largest entries of
the empirical mean of $Y$, whether errors or not, will be of great
importance.

These large entries will tell us where to focus our
attention. In particular, we can find the $k^2$ largest entries of the
empirical mean of $Y$ and attempt to filter based on them. 
When we do so, one of two things will happen: Either we 
remove bad samples and make progress or
we verify that these entries ought to be large, 
and thus must come from the support of $v$.
In particular, when we reach the second case, 
since the adversary cannot shrink the empirical variance of $v \cdot X$ by much, 
almost all of the entries on the support of $v$ must remain large, 
and this can be captured by our algorithm.

The above algorithm works under a set of 
deterministic conditions on the good set of samples
that are satisfied with high probability with 
$\poly(k)\log(d)/\eps^2$ samples. Our current analysis
does not establish the information-theoretically
optimal sample size of $O(k^2 \log(d)/\eps^2)$, though we believe that this 
plausible via a tighter analysis.

We note that a naive implementation of this algorithm will achieve error $\poly(\eps)$
in our final estimate for $v$, while our goal is to obtain $\tilde{O}(\eps)$ error.
To achieve this, we need to overcome two difficulties:
First, when trying to filter $Y$ on subsets of its coordinates, we do not
know the true variance of $Y$, and thus cannot expect to obtain $\tilde{O}(\eps)$
error. This is fixed with a bootstrapping method similar to
that in~\cite{Kane18} to estimate the covariance of a Gaussian. 
In particular, we do not know $\Var[Y]$ a priori, but after we
run the algorithm, we obtain an approximation to $v$, which gives an
approximation to $\Var[Y]$. This in turn lets us get a better
approximation to $v$ and a better approximation to $\Var[Y]$; and so on.

\section{Preliminaries} \label{sec:prelims}
We will use the following notation and definitions.


\paragraph{Basic Notation}
For $n \in \Z_+$, let $[n] \eqdef \{1, 2, \ldots, n\}$.
Throughout this paper, for $v = (v_1, \ldots, v_d) \in \R^d$, we will use $\| v \|_2$ to denote its Euclidean norm.
If $M \in \R^{d \times d}$, we will use $\| M \|_2$ to denote its spectral norm, $\| M \|_F$ to denote its Frobenius norm,
and $\tr[M]$ to denote its trace.
We will also let $\preceq$ and $\succeq$ denote the PSD ordering on matrices.
For  a finite multiset $S$, we will write $X \in_u S$ to denote that $X$ is drawn 
from the empirical distribution defined by $S$. Given finite multisets $S$ and $S'$ we let $\Delta(S,S')$
be the size of the symmetric difference of $S$ and $S'$ divided by the cardinality of $S$.

For $v \in \R^d$ and $S \subseteq [d]$, let $v_S$ be the vector with $(v_S)_i = v_i$, $i \in S$, 
and $(v_S)_i = 0$ otherwise. We denote by $h_k(v)$ the thresholding operator 
that keeps the $k$ entries of $v$ with largest magnitude (breaking ties arbitrarily) and sets the rest to $0$.
For $M \in \mathbb{R}^{d \times d}$ and $U \subseteq [d]$, let $M_{U}$ 
denote the matrix $M$ restricted to the $U \times U$ sub-matrix. 
For $W \subseteq [d] \times [d]$, then we will use $M_{(W)}$ to denote the matrix $M$ 
restricted to the elements whose entries are in $W.$

Let $\delta_{ij}$ denote the Kronecker delta function. 
We will denote $\erfc(z) = (2/\sqrt{\pi}) \int_{z}^{\infty} e^{-t^2} dt$.
The notation $\widetilde{O}(\cdot)$ and $\widetilde{\Omega}(\cdot)$ hides logarithmic factors
in the argument.

\section{Robust Sparse Mean Estimation} \label{app:sparse-mean}

\begin{algorithm}
	\begin{algorithmic}[1]
		\Procedure{Robust-Sparse-Mean}{$S, k, \eps,\tau$}
		\INPUT A multiset $S$ such that there exists an $(\eps,k, \tau)$-good set $G$ with $\Delta(G, S) \le 2\eps$.
		\OUTPUT Multiset $S'$ or vector $\wh{\mu}$ satisfying Proposition~\ref{prop:rsm}.
		\State Compute the sample mean $\wt{\mu}=\E_{X\in_u S}[X]$ and the sample covariance matrix $\wt{\Sigma}$ ,
		i.e., $\wt{\Sigma} = (\wt{\Sigma}_{i,j})_{1 \le i, j \le d}$ with 
		$\wt{\Sigma}_{i,j} = \E_{X\in_u S}[(X_i-\wt{\mu}_i) (X_j-\wt{\mu}_j)]$.
		\State Let $U \subseteq [d] \times [d]$ be the set of the $k$ largest magnitude entries of the diagonal of $\wt{\Sigma}-I$ 
		and the largest magnitude $k^2-k$ off-diagonal entries, with ties broken so that if $(i, j) \in U$ then $(j,i) \in U$. 
		\If {$\| (\wt{\Sigma}-I)_{(U)}\|_F \leq O(\eps \log(1/\eps))$} \label{step:return-mean}
		\textbf{return}  $\wh{\mu}: = h_k(\wt{\mu}).$
		\EndIf
		\State Set $U' = \{i \in [d]: (i, j) \in U\}$. \label{step:linear-filter-start}
		\State Compute the largest eigenvalue $\lambda^{\ast}$ of $(\wt{\Sigma}-I)_{U'}$ 
		and a corresponding unit eigenvector $v^{\ast}$.
		\If{$\lambda^{\ast} \geq \Omega(\eps \sqrt{\log(1/\eps)})$}: Let $\delta_{\ell} := 3\sqrt{\eps \lambda^{\ast}}$.
		Find $T > 0$ such that
		\[ \pr_{X \in_u S} \left[|v^{\ast} \cdot (X - \wt{\mu})  | \geq T + \delta_{\ell}\right] 
		\geq 9 \cdot \erfc(T/\sqrt{2})+ \frac{3 \eps^2}{T^2\ln(k \ln(N d/\tau))}. \]                             
		\State \textbf{return} the multiset $S' = \{x\in S: |v^{\ast} \cdot (x-\wt{\mu}) | \leq T+\delta_{\ell} \}$.
		\label{step:linear-filter-end}
		\EndIf
		
		\State Let $p(x)=\left((x-\wt{\mu})^T (\wt{\Sigma}-I)_{(U)}^T(x-\wt{\mu})-\Tr((\wt{\Sigma}-I)_{(U)} )\right)/\|(\wt{\Sigma}-I)_{(U)}\|_F$.
		\label{step:quad-filter-start}
		
		\State \label{step:quad-filter}  
		Find $T>6$ such that $$\pr_{X \in_u S} [|p(X)| \geq T] \geq  9 \exp(-T/4) +  3\eps^2/(T \ln^2 T) \;.$$
		\State \textbf{return} the multiset $S' = \{x \in S: |p(x)| \leq T \}$.	
		\label{step:quad-filter-end}
		\EndProcedure
	\end{algorithmic}
	\caption{Robust Sparse Mean Estimation via Iterative Filtering}
	\label{alg:rsm}
\end{algorithm}

In this section, we prove correctness of Algorithm~\ref{alg:rsm}
establishing Theorem~\ref{thm:mean-informal}.
For completeness, we restate a formal version
of this theorem:

\begin{theorem} \label{thm:mean-gaussian-identity}
Let $D \sim \normalpdf(\mu,I)$ be an identity covariance Gaussian distribution on $\R^d$ with
unknown $k$-sparse mean vector $\mu$, and $\eps, \tau > 0$.
Let $S$ be an $\eps$-corrupted set of samples from $D$ of size 
$N=\todo{\widetilde{\Omega}(k^2 \log(d/\tau)/\eps^2)}$. There exists an efficient algorithm that, 
on input $S$, $k$, $\eps$, and $\tau$, returns a mean vector $\widehat{\mu}$ 
such that with probability at least $1-\tau$ it holds $\|\widehat{\mu}-\mu\|_2 = O(\eps\sqrt{\log(1/\eps)}).$
\end{theorem}


\subsection{Proof of Theorem~\ref{thm:mean-gaussian-identity}}

In this section, we describe and analyze our algorithm establishing Theorem~\ref{thm:mean-gaussian-identity}.
We start by formalizing the set of deterministic conditions on the good data
under which our algorithm succeeds:

\begin{definition} \label{def:good-set-mean} 
Fix $0< \eps, \tau <1$ and $k \in \Z_+$.
A multiset $G$ of points in $\R^d$ is {\em $(\eps,k,\tau)$-good with respect to $\normalpdf(\mu,I)$} if, 
for $X \in_u G$ and $Y \sim \normalpdf(\mu,I)$, the following conditions hold:

\begin{itemize}
\item[(i)] For all $i \in [d]$, $|\E_{X \in_u G}[X_i]-\mu_i| \leq \eps/k$, 
              and for all $i, j \in [d]$, $|\E_{X \in_u G}\left[(X_i-\mu_i)(X_j-\mu_j)\right] - \delta_{ij}| \leq \eps/k$.
\item[(ii)] For all $x \in G$ and $i \in [d]$, we have $|x_i - \mu_i| \leq O(\sqrt{\log(d |G|/\tau)})$. 
\item[(iii)] For all $2k^2$-sparse unit vectors $v \in \R^d$, we have that: 
\begin{itemize}
\item[(a)] $|\E_{X\in_u G}[v \cdot (X-\mu)]| \leq O(\eps)$,
\item[(b)] $|\E_{X \in_u G}[(v \cdot (X-\mu))^2]-1| \leq O(\eps)$, and 
\item[(c)] For all $T \geq \todo{6}$, 
$\pr_{X \in_u G}[|v \cdot (X - \mu)| \geq T] \leq 3 \cdot \erfc(T/\sqrt{2}) + \eps^2/\left(T^2\ln\left(k \ln(d|G|/\tau)\right)\right).$
\end{itemize}
\item[(iv)] For all homogeneous\footnote{Recall that a degree-$d$ polynomial is called homogeneous if its non-zero terms
are all of degree exactly $d$.} degree-$2$ polynomials $p$ with $\Var_{\normalpdf(\mu, I)}[p(Y)]=1$ and at most $k^2$ terms, 
we have that:
\begin{itemize}
\item[(a)] $|\E_{X \in_u G}[p(X)] - \E_{\normalpdf(\mu, I)}[p(Y)]| \leq O(\eps \sqrt{\Var_{\normalpdf(\mu, I)}[p(Y)]}) = O(\eps)$, and, 
\item[(b)] For all $T \geq 5$, $\pr_{X \in_u G}\left[\left| p(X) - \E_{\normalpdf(\mu, I)}[p(Y)] \right| \geq T \right] \leq  3 \exp(-T/4) + \eps^2/(T \ln^2 T).$
\end{itemize}
\end{itemize}
\end{definition}
 
Our first lemma says that a sufficiently large set of samples from $\normalpdf(\mu,I)$
is good with high probability:
 
\begin{lemma} \label{lem:samples-good-mean} 
A set of $N=\todo{\widetilde{\Omega} \left(k^2 \log(d/\tau)/\eps^2 \right)}$ samples 
from $\normalpdf(\mu,I)$ is $(\eps,k,\tau)$-good (with respect to $\normalpdf(\mu,I)$) 
with probability at least $1-\tau$.
\end{lemma}

Our algorithm iteratively applies the procedure \textsc{Robust-Sparse-Mean} 
(Algorithm~\ref{alg:rsm}). 
The crux of the proof is the following performance guarantee of \textsc{Robust-Sparse-Mean}:

\begin{proposition} \label{prop:rsm}
Algorithm~\ref{alg:rsm} has the following performance guarantee: 
On input a multiset $S$ of $N$ points in $\R^d$ such that $\Delta(G, S) \leq 2\eps$, 
where $G$ is an $(\eps, k, \tau)$-good set with respect to $\normalpdf(\mu,I)$, 
procedure \textsc{Robust-Sparse-Mean} returns one of the following:
\begin{enumerate}
\item A mean vector $\wh{\mu}$ such that $\|\wh{\mu}-\mu\|_2 = O(\eps\sqrt{\log(1/\eps)})$, or 
\item A multiset $S' \subset S$ satisfying $\Delta(G, S') \leq \Delta(G, S) - \eps/N$.
\end{enumerate} 
\end{proposition}

\todo{
We note that our overall algorithm terminates after at most $2N$ iterations of Algorithm~1, in which
case it returns a candidate mean vector satisfying the first condition of Proposition~\ref{prop:rsm}.
Note that the initial $\eps$-corrupted set $S$ satisfies $\Delta(G, S) \le 2\eps$.
If ${S}^{(i)} \subset S$ is the multiset returned after the $i$-th iteration, then  
we have that $0 \leq \Delta({S}^{(i)}, G) \leq 2\eps-i (\eps/N)$.} 

In the rest of this section, we prove Proposition~\ref{prop:rsm}.

We start by showing the first part of Proposition~\ref{prop:rsm}.
Note that Algorithm~1 outputs a candidate mean vector only if 
$\|(\wt{\Sigma}-I)_{(U)}\|_F \leq O(\eps \log(1/\eps))$.
We start with the following lemma:

\begin{lemma} \label{lem:small-norm-1} 
If $\|(\wt{\Sigma}-I)_{(U)}\|_F \leq O(\eps \log(1/\eps))$, then for any $T \subseteq [d]$ with $|T| \leq k$,
we have $\|\wt{\mu}_T - \mu_T\|_2 \leq O(\eps \sqrt{\log(1/\eps))}$.
\end{lemma}
\begin{proof}
Fix $T \subseteq [d]$ with $|T| \leq k$.
By definition, $\|(\wt{\Sigma}-I)_T\|_F$ is the Frobenius norm of the corresponding sub-matrix on $T \times T$. 
Note that this is the $\ell_2$-norm of a set of $k$ diagonal entries and 
$k^2-k$ off-diagonal entries of $\wt{\Sigma}-I$. By construction, $U$ is the set that maximizes this norm, 
and therefore 
\[\|(\wt{\Sigma}-I)_T\|_2 \leq  \|(\wt{\Sigma}-I)_T\|_F \leq \|(\wt{\Sigma}-I)_{(U)}\|_F 
\leq O(\eps \log(1/\eps)) \;.\]
Given this bound, we leverage a proof technique from~\cite{DKKLMS16} 
showing that a bound on the spectral norm of the covariance implies a $\ell_2$-error 
bound on the mean. \todo{
This implication is not explicitly stated in~\cite{DKKLMS16}, but follows
directly from the arguments in Section~5.1.2 of that work. In particular, the analysis
of the ``small spectral norm'' case in that section shows that
$\|\wt{\mu}_T - \mu_T\|_2 \leq O\left(\sqrt{\eps} \|(\wt{\Sigma}-I)_T\|_2^{1/2} +\eps\sqrt{\log(1/\eps)}\right)$,
from which the desired claim follows.}
This completes the proof of Lemma~\ref{lem:small-norm-1}.
\end{proof}
 
Given Lemma~\ref{lem:small-norm-1}, the correctness of the sparse mean approximation 
output in Step~\ref{step:return-mean} of Algorithm~\ref{alg:rsm} follows from the following corollary:
 
\begin{corollary} \label{cor:small-norm-sparse}
Let $\wh{\mu} = h_k (\wt{\mu})$. 
If $\|(\wt{\Sigma}-I)_{(U)}\|_F \leq O(\eps \log(1/\eps))$, 
then  $\|\wh{\mu}-\mu\|_2 \leq O(\eps \sqrt{\log(1/\eps)})$. 
\end{corollary}
\begin{proof}
For vectors $x, y$,  let $N_x$ denote the set of coordinates on which $x$ is non-zero 
and $N_{x \mid y}$ denote the set of coordinates on which $x$ is non-zero and $y$ is zero. 
Setting $T = N_{\mu}$ and $T = N_{\wh{\mu} \mid \mu}$ in Lemma~\ref{lem:small-norm-1}, 
we get that $\|\wt{\mu}_{N_{\mu}} - \mu\|_2 \leq O(\eps \sqrt{\log(1/\eps)})$ 
and  $\|\wt{\mu}_{N_{\wh{\mu} |  \mu}}\|_2 \leq O(\eps \sqrt{\log(1/\eps)})$.

If $\wt{\mu}$ has $k$ or fewer non-zero coordinates, then $\wh{\mu}=\wt{\mu}$ 
and $\|\wh{\mu}-\mu\|_2 = \|\wt{\mu}_{N_{\mu} \cup N_{\wh{\mu} \mid  \mu}} - \mu\|_2 \leq O(\eps \sqrt{\log(1/\eps)})$ 
and we are done. Otherwise, $\wh{\mu}$ has exactly $k$ non-zero coordinates and so 
$|N_{\mu \mid \wh{\mu}}| \leq |N_{\wh{\mu} \mid \mu}|$.
Since the nonzero coordinates of $\wh{\mu}$ are the $k$ largest magnitude coordinates of $\wt{\mu}$, 
for any $i \in N_{\mu \mid \wh{\mu}}$ and $j \in N_{\wh{\mu} \mid  \mu}$,  we have that 
$|\wt{\mu}_i| \leq |\wt{\mu}_j|$. 
Since $\|\wt{\mu}_{N_{\wh{\mu} \mid \mu}}\|_2 \leq O(\eps \sqrt{\log(1/\eps)})$, 
at least one coordinate $j \in N_{\wh{\mu} \mid \mu}$ must have 
$\wt{\mu}_j^2 \leq O(\eps^2\log(1/\eps))/|N_{\wh{\mu} \mid  \mu}|$.
Therefore, for any $i \in N_{\mu \mid \wh{\mu}}$, we have that $\wt{\mu}_i^2 \leq O(\eps^2\log(1/\eps))/|N_{\wh{\mu} \mid  \mu}|$.

Thus, we have 
\[\|\wt{\mu}_{N_{\mu \mid \wh{\mu}}}\|_2^2 = 
\sum_{i \in N_{\mu \mid \wh{\mu}}} \wt{\mu}_i^2 \leq 
\frac{|N_{\mu \mid \wh{\mu}}| \cdot O(\eps^2 \log(1/\eps))}{|N_{\wh{\mu} \mid \mu}|}  
\leq O(\eps^2 \log(1/\eps)) \;,\]
where the second inequality used that $|N_{\mu \mid \wh{\mu}}| \leq |N_{\wh{\mu} \mid \mu}|$.

Since $\|\wt{\mu}_{N_{\mu}} - \mu\|_2 \leq O(\eps \sqrt{\log(1/\eps)})$, 
by the triangle inequality we have that 
$\|\mu_{N_{\mu \mid \wh{\mu}}}\|_2 \leq O(\eps \sqrt{\log(1/\eps)})$. 
Finally, we have that 
\[\|\mu-\wh{\mu}\|_2^2 = \|\mu_{N_{\mu} \cap N_{\wh{\mu}}} - \wh{\mu}_{N_{\mu} \cap N_{\wh{\mu}}}\|_2^2 
+ \|\mu_{N_{\mu \mid \wh{\mu}}}\|_2^2 + \|\wt{\mu}_{N_{\wh{\mu} \mid \mu}}\|_2^2 
\leq O(\eps^2 \log(1/\eps)) \;,\]
concluding the proof.
\end{proof}
 
\noindent Lemma~\ref{lem:small-norm-1} and Corollary~\ref{cor:small-norm-sparse} 
give the first part of Proposition~\ref{prop:rsm}. 

\medskip

We now analyze the complementary case that $\|(\wt{\Sigma}-I)_{(U)}\|_F  = \Omega(\eps \log(1/\eps))$.
In this case, we apply two different filters, a linear filter (Steps~\ref{step:linear-filter-start}-\ref{step:linear-filter-end}), 
and a quadratic filter (Steps~\ref{step:quad-filter-start}-\ref{step:quad-filter-end}). 
To prove the second part of Proposition~\ref{prop:rsm},
we will show that at least one of these two filters: 
(i) removes at least one point, and (ii) it removes more corrupted than uncorrupted points.

\todo{
The analysis in the case of the linear filter follows by a reduction to the linear filter
in~\cite{DKKLMS16} for the non-sparse setting (see Proposition~5.5 in Section 5.1 of that work). 
More specifically, the linear filter in Steps~\ref{step:linear-filter-start}-\ref{step:linear-filter-end}
is essentially identical to the linear filter of \cite{DKKLMS16} restricted to the $2k^2 \times 2k^2$
matrix $\wt{\Sigma}_{U'}$. We note that Definition~\ref{def:good-set-mean} 
implies that every restriction to $2k^2$ coordinates satisfies the properties of the good set 
in the sense of~\cite{DKKLMS16} (Definition 5.2(i)-(ii) of that work). This implies that the analysis
of the linear filter from~\cite{DKKLMS16} holds in our case, establishing the desired properties.
Since the linear filter removes more corrupted points than uncorrupted points, it 
will remove at most a $2\eps$ fraction of the points over all the iterations. 
}

If the condition of the linear filter does not apply, i.e., if $\| (\wt{\Sigma}-I)_{U'} \|_2  \leq O(\eps \log(1/\eps))$,
the aforementioned analysis in~\cite{DKKLMS16} implies
 $\|\wt{\mu}_{U'} - \mu_{U'}\|_2 \leq O(\eps\sqrt{\log(1/\eps)})$.
In this case, we show that the second filter behaves appropriately.

Let $p(x)$ be the polynomial considered in the quadratic filter.
We start with the following technical lemma analyzing 
the expectation and variance of $p(x)$ under various distributions:

\begin{lemma} \label{lem:exp-var-P} 
The following hold true:
\begin{itemize} 
\item[(i)] $\E_{\normalpdf(\wt{\mu},I)}[p(Y)]=0$ and $\Var_{\normalpdf(\wt{\mu},I)}[p(Y)]=1$.
\item[(ii)] $\E_{S}[p(X)]=\|(\wt{\Sigma}-I)_{(U)}\|_F $.
\item[(iii)]  $|\E_{\normalpdf(\mu,I)}[p(Z)]| \leq O(\eps^2 \log(1/\eps))$ 
and $\Var_{\normalpdf(\mu,I)}[p(Z)] = 1 + O(\eps^2 \log(1/\eps))$.
\end{itemize}
\end{lemma}
\begin{proof}
Let $A=\frac{(\wt{\Sigma}-I)}{\|(\wt{\Sigma}-I)\|_F}$ and 
$p(x):=(x-\wt{\mu})^T A_{(U)} (x-\wt{\mu}) - \Tr[A_{(U)}]$. 
We have
\begin{align*}
\E_{\normalpdf(\wt{\mu},I)}[(Y-\wt{\mu})^T A_{(U)} (Y-\wt{\mu})]&=\Tr[A_{(U)} \E_{\normalpdf(\wt{\mu},I)}[(Y-\wt{\mu})(Y-\wt{\mu})^T]] =\Tr[A_{(U)}I] =\Tr[A_{(U)}] \;.
\end{align*}
Therefore, $\E_{\normalpdf(\wt{\mu},I)}[p(Y)] = \Tr[A_{(U)}] -\Tr[A_{(U)}]=0$. 
Similarly,
\begin{align*}
\E_{S}[(X-\wt{\mu})^T A_{(U)} (X-\wt{\mu})] =\E_{S}[\Tr[A_{(U)}(X-\wt{\mu})(X-\wt{\mu})^T]] 
=\Tr[A_{(U)} \E_{S}[(X-\wt{\mu})(X-\wt{\mu})^T]]\ =\Tr[A_{(U)}\wt{\Sigma}] \;,
\end{align*}
and so
\begin{align*}
\E_{S}[p(X)]&=\Tr[A_{(U)}(\wt{\Sigma}-I)] =\Tr[A_{(U)} A] \|(\wt{\Sigma}-I)_{(U)}\|_F 
=\|A_{U}\|_F  \|(\wt{\Sigma}-I)_{(U)}\|_F = \|(\wt{\Sigma}-I)_{(U)}\|_F \;.
\end{align*}
We have thus shown (ii) and the first part of (i). 

We now proceed to show the first part of (iii). 
Note that 
\begin{align*}
\E_{\normalpdf(\mu,I)}[(Z-\wt{\mu})(Z-\wt{\mu})^T]&=\E_{\normalpdf(\mu,I)}[(Z-\mu)(Z-\mu)^T]+\E_{\normalpdf(\mu,I)}[(\wt{\mu}-\mu)(Z-\wt{\mu})^T]+\E_{\normalpdf(\mu,I)}[(Z-\mu)(\wt{\mu}-\mu)^T]\\
&=I+(\wt{\mu}-\mu)(\wt{\mu}-\mu)^T+0.
 \end{align*} 
Thus, we can write
\[\E_{\normalpdf(\mu,I)}[p(Z)]=\Tr[A_{(U)} (\E_{\normalpdf(\mu,I)}[(Z-\wt{\mu})(Z-\wt{\mu})^T]-I)] = 
(\wt{\mu}-\mu)^T A_{(U)} (\wt{\mu}-\mu) \;,\] 
and so 
\[|\E_{\normalpdf(\mu,I)}[p(Z)]| \leq \|\wt{\mu}_{\todo{U'}}-\mu_{\todo{U'}}\|_2^2 \|A_{(U)}\|_2 \leq \|\wt{\mu}_{\todo{U'}}-\mu_{\todo{U'}}\|_2^2 
\leq O(\eps^2 \log(1/\eps)) \;.\]
This proves all the statements about expectations. 

We now analyze the variance of $p(x)$ for $Y$ and $Z$. 
Since $A_{(U)}$ is symmetric, we can write $A_{(U)} = O^T \Lambda O$ for an orthogonal matrix $O$ 
and a diagonal matrix $\Lambda$. 
Note that $Y'=O (Y-\wt{\mu})$ is distributed as $\normalpdf(0,I)$. 
Under these substitutions, $p(Y) = \sum_i \Lambda_{ii} Y'^2_i - \sum \Lambda_{ii}$, the variance of which is the same as the variance of $p'(Y) =  \sum_i \Lambda_{ii} Y'^2_i$
and so \[\Var_{\normalpdf(\tilde{\mu}, I)}[p'(Y)]=\sum_i \Lambda_{ii}^2 \Var_{\normalpdf(0, I)}[Y'^2_i]= \|\Lambda\|_F^2 = \|A_{(U)}\|_F^2=1.\] 
 
Similarly, to estimate the variance of $p(Z)$ we see we just need to estimate the variance of $p'(Z) = (Z-\wt{\mu})^T O^T \Lambda O(Z-\wt{\mu}) = (Z- \mu + \mu -\wt{\mu})^T O^T \Lambda O(Z- \mu + \mu - \wt{\mu})$ where $(Z - \mu) \sim N(0, I)$. Let $Z' = O^T(Z - \mu)$ and $\nu = O(\mu - \wt{\mu})$.  

This gives us 
\begin{align*}
\Var_{\normalpdf(\mu,I)}[p'(Z)] &= \E_{\normalpdf(0,I)}[ (Z'\Lambda Z' - \E_{\normalpdf(0,I)}[Z'\Lambda Z' ] + 2 \nu^T \Lambda (Z' - \nu)^2] \\
&= \E_{\normalpdf(0,I)}[ (Z'\Lambda Z' - \E_{\normalpdf(0,I)}[Z'\Lambda Z' ])^2 ] + 4 \E_{\normalpdf(0,I)}[(\nu^T \Lambda (Z' - \nu))^2] \\
&+  4 \E_{\normalpdf(0,I)}[(\nu^T \Lambda (Z' - \nu))(Z'\Lambda Z' - E[Z'\Lambda Z' ])]\\
&\leq 1 + 8 \E_{\normalpdf(0,I)}[(\nu^T \Lambda (Z' - \nu))^2]\\
&\leq  1 + 8 \nu^T \Lambda^2 \nu + 8 (\nu^T \Lambda \nu)^2 \\
&\leq 1 + 8 \| \wt{\mu}_{U'} - \mu_{U'} \|_2^2 \| \Lambda^2 \|_2 + 8 \| \wt{\mu}_{U'} - \mu_{U'} \|_2^4 \| \Lambda \|^2_2 \\
&\leq 1 + O(\epsilon^2 \log(1/\epsilon))
\end{align*}

%
%
%
This completes the proof of Lemma~\ref{lem:exp-var-P}.
\end{proof}

\todo{
Suppose that we find a threshold $T>0$ such that Step~\ref{step:quad-filter}
of the algorithm holds, i.e., the quadratic filter applies. Then we can show that 
Step~\ref{step:quad-filter-end} removes more bad points than good points.
This follows from standard arguments, by combining Definition~\ref{def:good-set-mean}(iv)(b) with our upper bound
for $\E_{\normalpdf(\mu,I)}[p(Z)]$ from Lemma~\ref{lem:exp-var-P}.
Let $S = G \cup E \setminus L$. 
By Definition~\ref{def:good-set-mean}(iv)(b), for the good set $G$, we have that for $T > 4$, 
 $\pr_{X \in_u G}\left[\left| p(X) - \E_{\normalpdf(\mu,I)}[p(Z)] \right| \geq T \right] \leq  3 \exp(-T/4) + \eps^2/(T \ln^2 T)$.
Lemma~\ref{lem:exp-var-P}(iii) implies that $|\E_{\normalpdf(\mu, I)}[p(Z)]| \leq O(\eps^2 \log(1/\eps))$.  Therefore, shifting $T$ by $\epsilon^2 \log(1/\eps)$ we obtain the following corollary:
\begin{corollary} \label{cor:p-G-sensible} We have that:
\begin{enumerate}
\item[(i)] $|\E_{X \in_u G}[p(X)]| \leq O(\eps)$ and, 
\item[(ii)] For $T \geq 6$, $\pr_{X \in_u G}[|p(X)| \geq T] \leq  (3+O(\eps)) \exp(-T/4) + (1+O(\eps))\left(\eps^2/(T \ln^2 T) \right)$. 
\end{enumerate}
\end{corollary}
Condition (ii) implies that the fraction of points in $G$ that violate the quadratic filter condition 
is less than $1/2$ the fraction of points in $S$ that violate the same condition. Therefore, 
the quadratic filter removes more corrupted points than uncorrupted ones.

}

It remains to show that if Algorithm~1 does not terminate in Step~\ref{step:return-mean}
and the linear filter does not apply, then the quadratic filter necessarily applies. 
To establish this, we need a couple more technical lemmas.
We first show that the expectation of $p(x)$ over 
the set of good samples that are removed is small:
\begin{lemma} \label{lem:pL-small}
We have that $|L| \cdot |\E_{X \in_u L}[p(X)]| \leq |S|  \cdot O(\eps \log(1/\eps))$. 
\end{lemma}
\begin{proof}
Since $L \subset G$ and $|G| = O(|S|)$, for $T \geq 6$ we have
\[|L| \cdot \pr_{X \in_u L}[|p(X)| \geq T] \leq  |G| \cdot \pr_{X \in_u G}[|p(X)| \geq T] \leq O\left(|S| (\exp(-T/4) + \eps^2/(T \ln^2 T))\right) \;,\]
where we used Corollary~\ref{cor:p-G-sensible}.
Thus, we obtain that for $\epsilon < O(1)$. 
\begin{align*}
|L| \cdot |\E_{X \in_u L}[p(X)]|  & \leq |L| \cdot \E_{X \in_u L}[|p(X)|] \\
&=  \int_0^{\infty} |L| \cdot \pr_{X \in_u L}[|p(X)| \geq T]  dT \\
&\leq \int_0^{3\ln(1/\eps)} |L| dT + \int_{3\ln(1/\eps)}^{\infty} O(|S| (\exp(-T/4) + \eps^2/(T \ln^2 T))) dT \\
&\leq O(|S|\eps\log(1/\eps)) + O(|S| \eps) + O(|S|\eps^2/\log \log(1/\eps)) \\
&= O(|S|\eps\log(1/\eps)) \;,
\end{align*}
where we used the fact that $|L| = O(\eps |S|)$ and that the derivative of $1/\ln x$ is $1/x \ln^2 x$.
This completes the proof of Lemma~\ref{lem:pL-small}.
\end{proof}
 
By a similar argument, we can show that if the quadratic filter does not apply, 
then the remaining points in $E$ contribute a small amount to the expectation of $p(x)$. 
 
\begin{lemma} \label{lem:pE-small} 
Suppose that for all \todo{$T\geq6$}, we have
$\pr_{X \in_u S} [|p(X)| \geq T] \leq  9 \exp(-T/4) + 3 \eps^2/(T \ln^2 T)$. Then, we have that $|E| \cdot |\E_{X \in_u E}[p(X)]| \leq O(|S| \eps \log(1/\eps))$.
\end{lemma}

By combining the above, we obtain the following corollary, completing 
the analysis of our algorithm:
 
\begin{corollary} 
If we reach Step~\ref{step:quad-filter} of Algorithm~\ref{alg:rsm}, then there exists a 
\todo{$T\geq6$} such that $\pr_{X \in_u S} [|p(X)| \geq T] \geq  9 \exp(-T/4) + 3 \eps^2/(T \ln^2 T)$.
\end{corollary}
\begin{proof}
Suppose for a contradiction that no such $T$ exists. 
Using Corollary~\ref{cor:p-G-sensible}, Lemmas \ref{lem:pL-small} and~\ref{lem:pE-small}, 
we obtain that 
\begin{align*}
|S| \cdot \|(\wt{\Sigma}-I)_U\|_F &= |S| \cdot \E_{S}[p(X)]
=|G| \cdot \E_{X \in_u G}[p(X)] + |E| \cdot \E_{X \in_u E}[p(X)]- |L| \cdot \E_{X \in_u L}[p(X)]\\
&=O(|S|\eps \log(1/\eps)) \;.
\end{align*} 
This is a contradiction, as if this was the case, Algorithm~\ref{alg:rsm} 
would have returned in Step~\ref{step:return-mean}.
\end{proof}





\section{Robust Sparse PCA} \label{app:pca}
\begin{algorithm}
	\begin{algorithmic}[1]
		\Procedure{Robust-Sparse-PCA}{$S, k, \widetilde{\Sigma}, \eps, \delta, \tau$}
		\INPUT A multiset S, an estimate of the true covariance $\widetilde{\Sigma}$, a real number $\delta \in \mathbb{R}$.
		\OUTPUT A multiset $S'$ or matrix $\Sigma'$ satisfying Proposition~\ref{prop:rspca}.
		\State For any $x \in \mathbb{R}^d$ define $\gamma(x) := \vvec(xx^T-I) \in \mathbb{R}^{d^2}$. 
		\State Compute $\tmu :=\E_{S}[\gamma(x)]$, $\hat{\mu} = h_{k^2}(\mu)$ and $Q := \supp(\hat{\mu})$. 
		\State Compute \[ M_{Q} := \E_{S}[(\gamma(x) - \tilde{\mu}) (\gamma(x)-\tilde{\mu})^T]_{Q\times Q} \in \mathbb{R}^{k^2} \times \mathbb{R}^{k^2}\]
		\State Let $\lambda, v^*$ be the maximum eigenvalue and corresponding eigenvector of $M_Q - \Cov_{X\sim \mathcal{N}(0,\wtSigma)}(\gamma(x)_Q)$.
		\If { $\lambda < C \cdot (\delta + \eps \log^2 ( 1/\eps))$, where $C$ is a sufficiently large constant}\label{step:small_lambda}
		\State Compute $w$, the largest eigenvector of $\mathsf{mat}(\tilde{\mu})_Q$. \Return $ww^T+I$.		\EndIf
		\State Let $\hat{\mu} = \mathsf{median}\left(\{ \gamma(x) \cdot v^* \mid x \in S \}\right)$.  Find a number $T > \log(1/\eps)$ satisfying
		\[ \pr_{S}[| \gamma(x)_Q \cdot  v^* - \hat{\mu} | > CT + 3] > \frac{\eps}{T^2 \log^2(T)}. \]
		
		\Return $S' = \{ x \in S \mid | (\gamma(x)_Q \cdot v^*) - \hat{\mu} | < T \}$.
		
		\EndProcedure
	\end{algorithmic}
	\caption{Robust Sparse PCA via Iterative Filtering}
	\label{alg:rspca}
\end{algorithm}

In this section, we prove correctness of Algorithm~\ref{alg:rspca}
establishing Theorem~\ref{thm:sparse-PCA-informal}, 
which we restate for completeness:
\begin{theorem} \label{thm:pca-app}
Let $D \sim \normalpdf(\mathbf{0}, I+ \rho v v^T)$ be a centered Gaussian distribution on $\R^d$
with spiked covariance $\Sigma = I+ \rho v v^T$ for an unknown $k$-sparse unit vector $v$,
and $0 < \rho < O(1)$ a real number. For some $\eps > 0$, let $S$ be an $\eps$-corrupted set of samples from $D$ of size
$N={\Omega}(k^4 \log^4(d/\eps)/\eps^2)$. There exists an algorithm that,
on input $S$, $k$, and $\eps$, runs in polynomial time and returns \new{$w \in \R^d$}
such that with probability at least $2/3$ we have that $\|ww^T - vv^T \|_F = O\left(\frac{\eps \log(1/\eps)}{\rho}\right)$. 

We will require some additional notation.
For any $M \in \mathbb{R}^{d \times d}$, define $\vvec(M) \in \mathbb{R}^{d^2}$ to be a canonical flattening of this vector and $\gamma(x) \in \mathbb{R}^{d^2}$ to be $\vvec(xx^T - I)$. Also let $\gamma_A(x) = \gamma(x_A)$ and $\vvec_B(M) = \vvec(M_B)$, where $A \subset [d] \times [d]$ and $x, M \in \mathbb{R}^{d \times d}$.  
\end{theorem}

As is standard with such robust statistics arguments, we will need to assume that the uncorrupted set of good samples $G$ has some desired properties. In particular, we will make use of the following notion of a good set:
\begin{definition}\label{def:good_rspca_boot}
	Define a set $G \subset \mathbb{R}^n$ to be $(\eps, k)$-good for $\mathcal{N}(0, I + \rho vv^T)$ and $\rho > 0$ if the following hold for every $Q \subset [d]\times [d]$
	\begin{enumerate}
		\item \label{cond:coord_dev} For some sufficiently large constant $C$ and for every $i\in [d]$ and $x\in G$, $|x_i| \leq C\sqrt{\log(d|G|)}$.
		\item \label{cond:mean} $\left \| (\E_{G} [xx^T] - I - \rho vv^T)_Q \right \|_F \leq \eps $
		\item \label{cond:var} For all $w \in \R^{k^2}$,
		\[ \Var_G[\gamma_Q(x) \cdot w] = (1 \pm \eps) \Var_{\mathcal{N}(0, I+\rho vv^T)}[\gamma_Q(x) \cdot w] \]
		\item \label{cond:concentration} For $C$ a sufficiently large constant, and for all $w \in \R^{k^2}$ satisfying $\|w \|_2 = 1$, and all $T>\log(1/\eps)$
		\[\pr_{G}[ | \gamma_Q(x) \cdot  w -  \rho \vvec_Q(vv^T) \cdot w |   > C T] < \frac{\eps}{T^2 \log^2(T)} \]
	\end{enumerate}
\end{definition}

We note that given a sufficiently large set of independent samples from $X$ that the above conditions hold with high probability.

\begin{lemma}\label{lem:samples-good-pca}
	If $G$ is a set of $N = C k^4 \log^4(d/\eps)/\eps^2$ samples drawn from $\mathcal{N}(0, I+ \rho vv^T)$, for $C$ a sufficiently large constant. Then $G$ is $(\eps, k)$-good with probability at least $2/3$.
\end{lemma}
We think in fact that we should be able to produce a good set with substantially fewer samples.
\begin{conjecture}
	There exists an $N=k^2 \mathrm{polylog}(d/\eps)/\eps^2$ so that if $G$ is a set of $N$ samples drawn from $\mathcal{N}(0, I+\rho vv^T)$, then $G$ is $(\eps, k)$-good with probability at least $1-1/d$.
\end{conjecture}

We can now proceed with the proof of our main Theorem. In particular, our algorithm will follow quickly from the existence of the following subroutine:
\begin{proposition}\label{prop:rspca}
	Let $G$ be an $(\eps,k)$-good set for $\mathcal{N}(0,\Sigma)$ with $\Sigma=I+ \rho vv^T$ with $v$ a unit length, $k$-sparse vector and $0 < \rho < 1$. There exists an algorithm (Algorithm~\ref{alg:rspca}) that given a matrix $\tilde \Sigma$ and a set $S$ with $\|\tilde \Sigma -\Sigma \|_F \leq \delta$ and $\Delta(S,G)\leq \eps |G|$ returns either a matrix $\Sigma'$ with $ \|\Sigma'-\Sigma \|_F = O(\sqrt{\eps\delta}+\eps\log(1/\eps))$ or a subset $T\subset S$ with $\Delta(T,G) < \Delta(S,G)$.
\end{proposition}
Our main theorem follows from iteratively applying the Proposition. The error stabilizes at $\delta $ with $\delta = O(\sqrt{\eps\delta}+\eps\log(1/\eps))$, which implies that $\delta= O(\eps\log(1/\eps))$. We begin by analyzing what happens when our algorithm returns a matrix. We first note that if we pass the filter, then $\tmu_Q$ will be approximately correct.
\begin{lemma}\label{QMeanCloseLem}
With the notation as in Algorithm~\ref{alg:rspca}, we have that $\|\tmu_Q - \vvec_Q(\Sigma-I)\|_2 = O(\sqrt{\eps\lambda} + \sqrt{\eps\delta}+\eps\log(1/\eps))$.
\end{lemma}
\begin{proof}
	Let $\|\tmu_Q - \vvec_Q(\Sigma-I) \|_2=a$ and $S=(G \backslash L)\cup E$. We wish to show that
	\[
	\left \|\sum_{x\in S} (xx^T-\Sigma)_Q \right \|_2 =  O(\sqrt{\eps\delta} + \eps\log(1/\eps))|G|.
	\]
	By the triangle inequality, the left hand side above is at most
	\begin{align}\label{eqn:mean_close_triangle}
	\left \|\sum_{x\in G} (xx^T-\Sigma)_Q \right \|_2+\left \|\sum_{x\in L} (xx^T-\Sigma)_Q \right \|_2+\left \|\sum_{x\in E} (xx^T-\Sigma)_Q \right \|_2.
	\end{align}
	Since $G$ is a good set, by Condition \ref{cond:mean}, we have that the first term is $O(\eps |G|)$. We now bound the second term. Since $\Sigma= I + \rho vv^T$ and $\gamma(x) = \vvec(xx^T - I)$ the second term is at most the supremum over unit vectors $w\in \R^{k^2}$ of
	\[
	\sum_{x\in L} (w\cdot \gamma_Q(x) - \rho w\cdot \vvec_Q(vv^T))
	\]
	Using the fact that for any random variable $\E_{D}[X] = \int_0^{\infty} \pr_{D}[X > t] dt$, this is at most
	\[
	\int_0^\infty \left|\{x\in L:|(w\cdot \gamma_Q(x) - \rho w\cdot \vvec_Q(vv^T)|>t \} \right|dt.
	\]
	Since $L \subset G$ and $|L| = \eps |G|$ this is at most
	\[
	\int_0^{C\log(1/\eps)} \eps|G| + \int_{C\log(1/\eps)}^\infty \eps/((t/C)^2\log^2(t/C))|G|dt = O(\eps\log(1/\eps) |G|),
	\]
	where the bound on the second term above is by Condition \ref{cond:concentration} of the definition of a good set. 
	
	We can bound the final term in~\ref{eqn:mean_close_triangle} by Cauchy-Schwartz as 
	\[
	(\eps|G|)^{1/2}\left(\sum_{x\in E} (w\cdot(xx^T-\Sigma)_Q)^2 \right)^{1/2}.
	\]
	
	To bound this  we note that 
	\[
	\sum_{x\in S} (w\cdot(xx^T-\Sigma)_Q)^2 \leq |S|(\var_S[w\cdot\gamma_Q(x)]+a^2).
	\]
	We also know that
	\[
	\var_G[w\cdot\gamma_Q(x)] = \var_{\mathcal{N}(0,\rho vv^T+I)}[w\cdot\gamma_Q(x)]+O(\eps) = \var_{\mathcal{N}(0,\wtSigma)}[w\cdot\gamma_Q(x)]+O(\eps+\delta).
	\]
	Thus subtracting both sides by $\var_{\mathcal{N}(0,\wtSigma)}[w\cdot\gamma_Q(x)]$ and scaling by $G$ gives us
	\[
	\sum_{x\in G} ((w\cdot(xx^T-\Sigma)_Q)^2-\var_{\mathcal{N}(0,\wtSigma)}[w\cdot\gamma_Q(x)] = O(\eps+\delta)|G|.
	\]
	However, since $v^*$ is the eigenvector corresponding to the largest eigenvalue, we also have that
	\[
	\var_S[w\cdot\gamma_Q(x)]-\var_{\mathcal{N}(0,\wtSigma)}[w\cdot\gamma_Q(x)] \leq \lambda+a^2.
	\]
	Combining with the above and using the fact that $S = (G \setminus L) \cup E$.  we have that
	\begin{align*}
	\sum_{x\in E} ((w\cdot(xx^T-\Sigma)_Q)^2 -\var_{\mathcal{N}(0,\wtSigma)}[w\cdot\gamma_Q(x)]) &\leq \sum_{x\in L} ((w\cdot(xx^T-\Sigma)_Q)^2 -\var_{\mathcal{N}(0,\wtSigma)}[w\cdot\gamma_Q(x)])\\ &+|G|O(\eps+\delta+\lambda+a^2).
	\end{align*}
	However, we can bound
	\[
	\sum_{x\in L} ((w\cdot(xx^T-\Sigma)_Q)^2 -\var_{\mathcal{N}(0,\wtSigma)}[w\cdot\gamma_Q(x)])
	\]
	by
	\[
	O(|L|)+\int_0^\infty \left|\{x\in L:|(w\cdot \gamma_Q(x) - \rho w\cdot \vvec_Q(vv^T)|>t \} \right|2t dt.
	\]
	As before, this is at most
	\[
	\int_0^{C\log(1/\eps)} 2t\eps|G|dt + \int_{C\log(1/\eps)}^\infty \eps/((t/C)^2\log^2(t/C))|G|2tdt = O(\eps\log^2(1/\eps) |G|).
	\]
	Thus, the final term in our sum is at most
	\[
	(\eps|G|)^{1/2} O(\eps\log^2(1/\eps)|G|+(a^2+\eps+\delta+\lambda)|G|)^{1/2}
	\]
	
	Therefore, we have that
	\[
	a=O(a\sqrt{\eps}+\sqrt{\eps\lambda} + \sqrt{\eps\delta}+\eps\log(1/\eps)),
	\]
	from which we conclude our result.
	
\end{proof}

Given this, we would like to show that $\Sigma'$ is close to $\Sigma$. In particular, we have:
\begin{lemma}\label{lem:vec_closeness}
	Suppose that $A=\E_{x\in_u S}[xx^T-I]$ and $Q$ the set of its $k^2$ largest entries. If $\|(A-\rho  vv^T)_Q\|_F = \eta$ then for $w$ a normalized, principle eigenvector of $A_Q$ we have that $\rho w$ is within $O(\eta+\eps\log(1/\eps))$ of either $\rho v$ or $-\rho v$.
\end{lemma}
Before we begin with the proof, we make an important observation:
\begin{lemma}
	In the notation above, for any set of entries $R$ defining a $k^2\times k^2$ submatrix, $(A+I)_R \geq ((\rho vv^T+I) - O(\eps\log(1/\eps))I)_R$, as self-adjoint operators.
\end{lemma}
\begin{proof}
	Note that $A+I = \E_{x\in_u S}[xx^T] \geq (1-\eps)\E_{x\in_u G\backslash L}[xx^T].$ By Property 2 of what it means to be a good set, $\E_{x\in_u G}[xx^T] = \rho vv^T+I+O(\eps)$. Thus, it suffices to show that for any unit vector $u$ with support of size at most $k^2$ that $|L|/|G|\E_{x\in_u L}[(x\cdot u)^2] = O(\eps\log(1/\eps)).$ This follows easily from Property 4.
\end{proof}

We are now ready to prove Lemma \ref{lem:vec_closeness}.
\begin{proof}
Let $R$ be the support of $vv^T$. Note that $A$ has larger total $L^2$ mass on $Q$ than it does on $R$. Therefore,
\[
\|A_{R\backslash Q} \|_F \leq \| A_{Q\backslash R} \|_F \leq \| (A-\rho vv^T)_Q \|_F = \eta.
\]

Let $B=(\rho vv^T)_{R\backslash Q}$. We note that with respect to Frobenius norm:
\[
A_Q = A_{Q\cap R} + A_{Q\setminus R} = A_{Q \cap R}+ O(\eta) = (\rho vv^T)_{Q\cap R} + O(\eta) = (\rho vv^T-B)+O(\eta).
\]

We also note that this is $A_{Q\cap R}+O(\eta) = A_R+O(\eta)$. Combining this with the above lemma, we have that
\[
(\rho vv^T-B+I)+O(\eta) \geq (\rho vv^T+I)-O(\eps\log(1/\eps))I.
\]

Rearranging, we find that $B\leq O(\eta+\eps\log(1/\eps))I$. But we note that the sign of the $i,j$ entry of $B$ is the same as the sign of $v_iv_j$ or 0. This means that $B$ is similar to a matrix with non-negative entries, and thus by The Perron--Frobenius Theorem, the largest eigenvalue of $B$ is positive, and hence $\|B \|_2 = O(\eta+\eps\log(1/\eps)).$ Therefore, we have that
\[
\|A_Q- \rho vv^T \|_2 \leq \|A_Q-( \rho vv^T-B) \|_2 + \|B \|_2 = O(\eta+\eps\log(1/\eps)).
\]
Note that unless $\eps$ and $\eta$ are sufficiently small, there is nothing to prove. Otherwise, we have that $v\cdot A_Q v \geq \rho - O(\eta + \eps \log(1/\eps))$, so $w$ will be an eigenvector with some eigenvalue $\lambda > \rho/2$. Since $\|A_Q- \rho vv^T \|_2 <\rho/2$, this means that $w$ must have a non-trivial component in the $v$-direction. Assume that $w$ is proportional to $v+u$ with $u$ orthogonal to $v$. Then we have that
\[
\lambda(v+u) = \lambda w = A_Q w = A_Q(v+u) = \rho v + O(\eta+\eps\log(1/\eps)).
\]
Taking the perpendicular to $v$ component above, we have that $\| u \|_2 = O(\eta+\eps\log(1/\eps))$, and this completes our proof.
\end{proof}

Finally, note that

\begin{align*}
\|vv^T - ww^T\|_F^2&=  \|vv^T \|_F^2 + \| ww^T\|_F^2 - 2\mathrm{tr}(vv^T ww^T) \\
&=2-2(v\cdot w)^2 \ll 2-2|v\cdot w|  =  \|v\pm w \|_2^2\\
&= \frac{ \|\rho v \pm \rho w \|_2^2 }{\rho^2}\leq O\left(\left( \frac{\eta + \eps \log(1/\eps)}{\rho}\right)^2\right)
\end{align*}



Thus, plugging in $\eta = O(\sqrt{\eps\lambda}+\sqrt{\eps\delta}+\eps\log(1/\eps))$ above, we find that $\|vv^T - ww^T \|_F= O(\frac{\sqrt{\eps\delta}+\eps\log(1/\eps)}{\rho})$.

We have left to analyze what happens when our algorithm returns a set $S'$. It is easy to see by Conditions \ref{cond:mean} and \ref{cond:var} that only $1/3$ of the elements of $G$ have $(\gamma(x)_Q-\rho \vvec(vv^T)_Q)\cdot v^*>3.$ Therefore, we have that $\hat{\mu}$ is within $3$ of $\rho v^*\cdot (\vvec(vv^T)_Q)$. From this and Condition \ref{cond:concentration} it is easy to see that if $C$ is sufficiently large (even compared to the $C$ in Condition \ref{cond:concentration}), that less than half of the elements of $S$ with $|\vvec(xx^T)\cdot v^*| > CT+3$ will be in $G$, and thus $\Delta(S',G) < \Delta(S,G)$.

All that remains is to show that such a threshold $T$ exists. To do this consider
\[
\var_S[v^*\cdot \vvec(xx^T)_Q]-\var_{\mathcal{N}(0,\Sigma)}[v^*\cdot \vvec(xx^T)_Q].
\]
This is $O(\delta)+\lambda$. On the other hand, since translating a random variable should not change it's variance, we see. 
\begin{align*}
\var_S(v^*\cdot \vvec(xx^T)_Q) &= \E_S[(v^*\cdot \vvec(xx^T-\rho vv^T)_Q)^2]-\E_S[v^*\cdot \vvec(xx^T-\rho vv^T)_Q]^2\\
& = \E_S[(v^*\cdot \vvec(xx^T-\rho vv^T)_Q)^2] + O(\eps\lambda+\eps\delta+\eps^2\log^2(1/\eps))
\end{align*}
by Lemma \ref{QMeanCloseLem}. Thus, 
\[
\E_S[(v^*\cdot \vvec(xx^T-\rho vv^T)_Q)^2] \geq \var_{\mathcal{N}(0,\Sigma)}[v^*\cdot \vvec(xx^T)_Q] + \lambda/2.
\]

Now by Conditions \ref{cond:mean} and \ref{cond:var} we have that
\[
\sum_{x\in G} ((v^*\cdot \vvec(xx^T-\rho vv^T)_Q)^2 - \var_{\mathcal{N}(0,\Sigma)}[v^*\cdot \vvec(xx^T)_Q]) = O(|G|\eps).
\]
By arguments from the proof of Lemma \ref{QMeanCloseLem}, we also have that
\[
\sum_{x\in L} ((v^*\cdot \vvec(xx^T- \rho vv^T)_Q)^2 - \var_{\mathcal{N}(0,\Sigma)}[v^*\cdot \vvec(xx^T)_Q)] = O(|G|\eps\log^2(1/\eps)).
\]
Thus, we must have
\[
\sum_{x\in E}(v^*\cdot \vvec(xx^T-\rho vv^T)_Q)^2 \gg |G|\lambda.
\]
However, this is at most
\[
O\left(|E|+\int_0^\infty\left|\{x\in E: |v^*\cdot \vvec(xx^T)_Q-\hat\mu >CT+3\} \right|tdt \right).
\]
If there is no such threshold, this is at most
\[
O\left(|E| +\int_0^{\log(1/\eps)} |E|tdt + \int_{\log(1/\eps)}^\infty \eps/(t^2\log^2(t)) tdt\right) = O(\eps\log^2(1/\eps)|G|),
\]
which is a contradiction. This completes our proof.



\section{Experiments} \label{sec:experiments}

For every experiment, we run 10 trials and plot the median value of the measurement. We shade the interquartile range around each measurement as a measure of the confidence of that measurement. 

Each experiment was run on a computer with a 2.7 GHz Intel Core i5 processor with an 8GB 1867 MHz DDR3 RAM.

\subsection{Robust Sparse Mean Estimation}

\newcommand{\RANSAC}{\texttt{RANSAC}}
\newcommand{\oracle}{\texttt{oracle}}
\newcommand{\NPsp}{\texttt{NP}}
\newcommand{\RMEsp}{\texttt{RME\_sp}}
\newcommand{\RMEnonsp}{\texttt{RME}}
\newcommand{\RMEspL}{\texttt{RME\_sp\_L}}
\newcommand{\RDPCA}{\texttt{RDPCA}}
\newcommand{\RSPCA}{\texttt{RSPCA}}

\begin{figure}
  \centering
  \begin{subfigure}[t]{0.45\textwidth}
    \includegraphics[width=\textwidth]{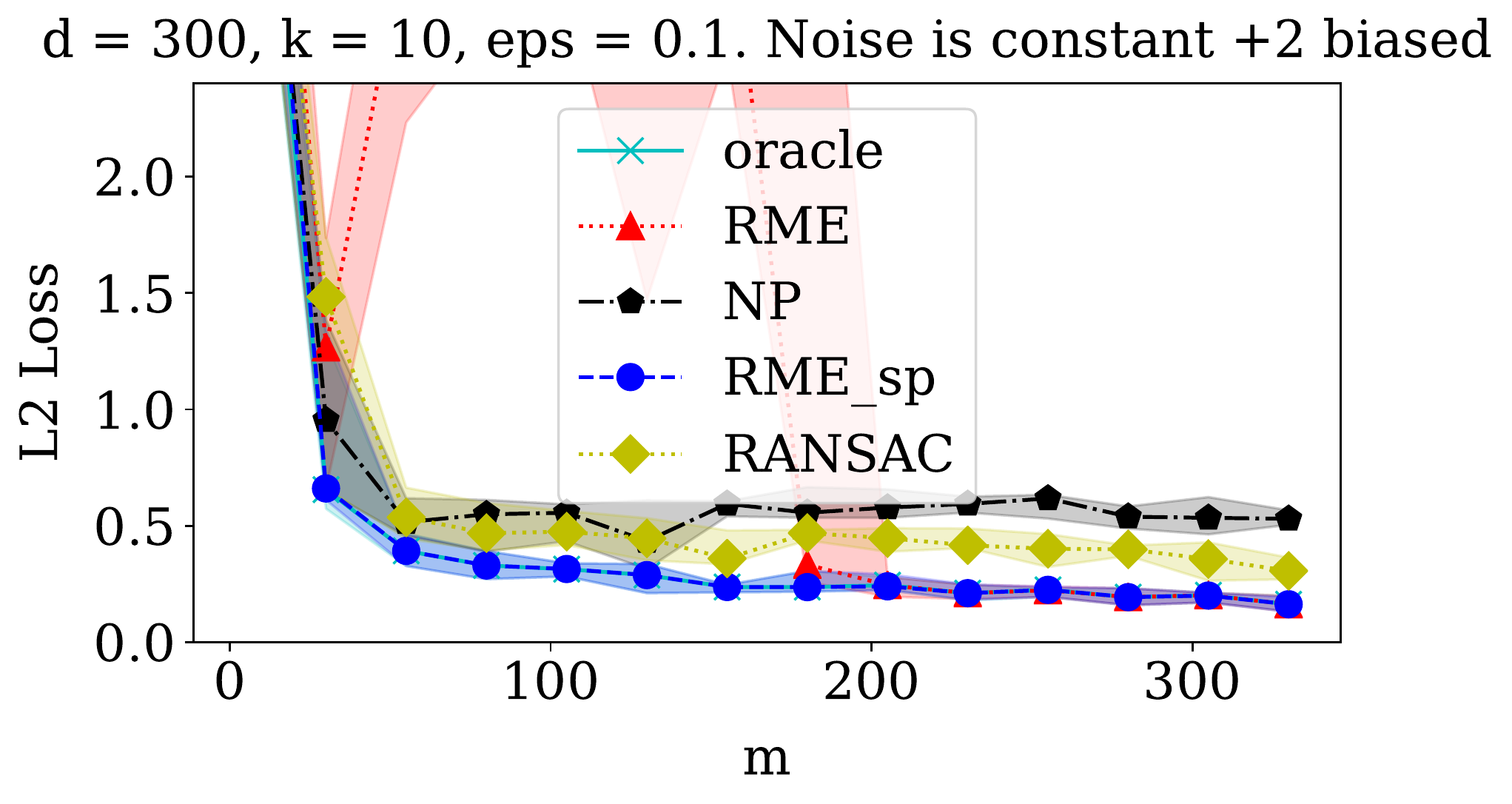}
    \caption{Unlike \RANSAC, our algorithm \RMEsp{} can filter out the noise and
      match the oracle's performance. \RMEnonsp{} also matches the oracle, but needs more samples.}
  \end{subfigure}
    ~ 
    \begin{subfigure}[t]{0.45\textwidth}
        \includegraphics[width=\textwidth]{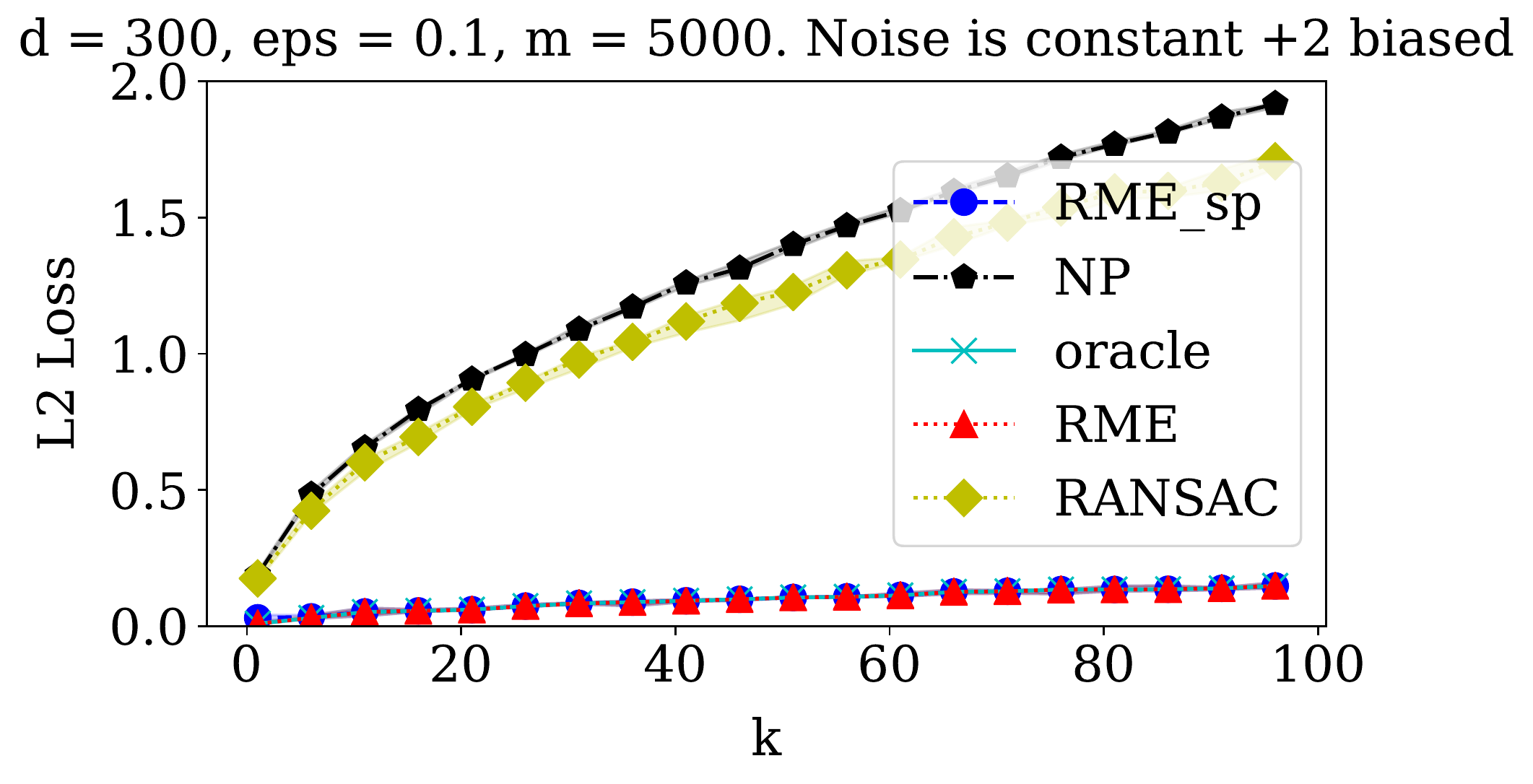}
        \caption{For fixed $m$, as $k$ increases, \RANSAC{} and \NPsp{} both diverge from \RMEnonsp{} and \RMEsp{}.}
    \end{subfigure}
    \caption{Constant-bias noise is easy for our algorithm, since it
      is caught by the linear filter.}
  \label{fig:const-bias}
  \centering
  \begin{subfigure}[t]{0.45\textwidth}
    \includegraphics[width=\textwidth]{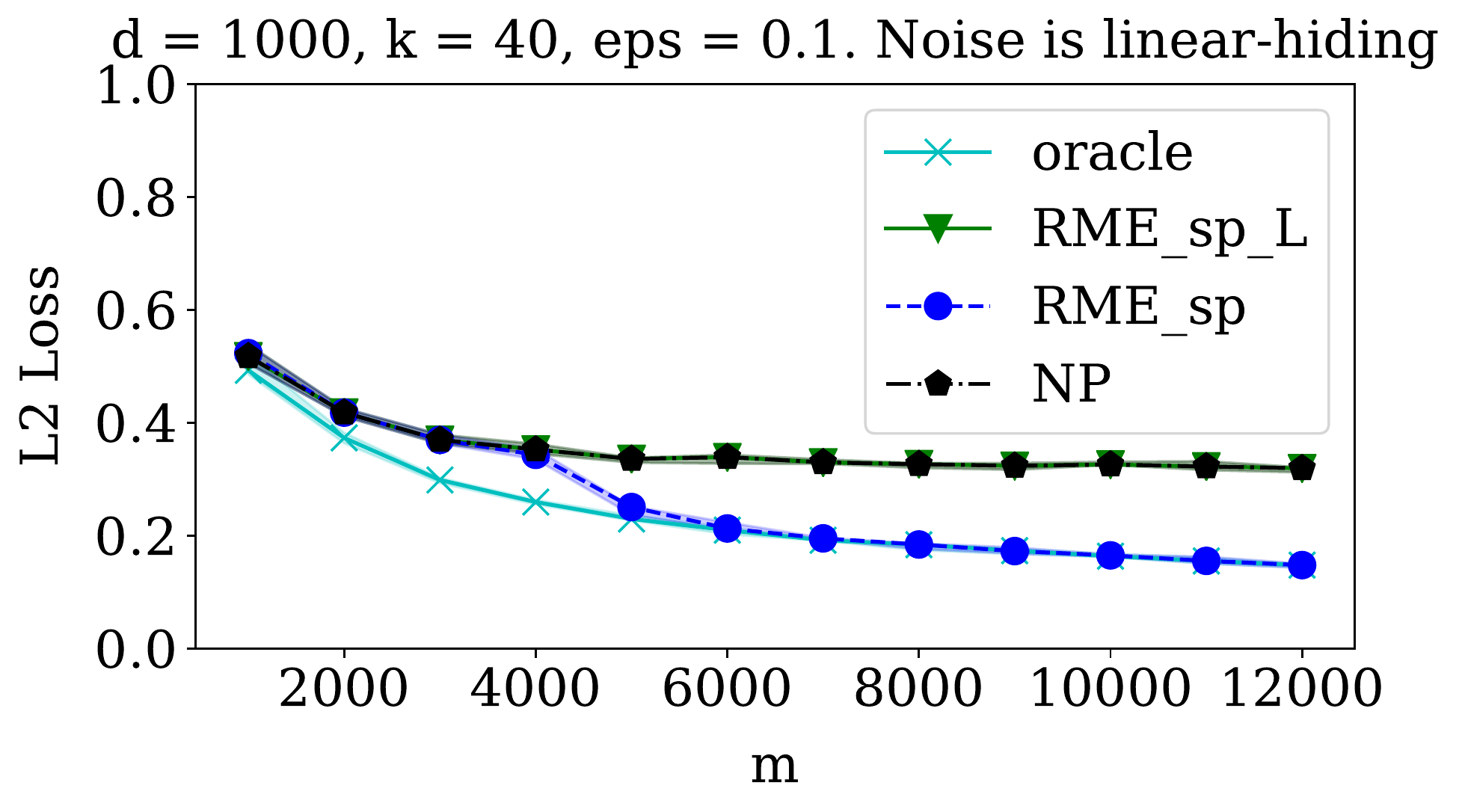}
    \caption{With sufficiently many samples, the quadratic filter can
      filter out the noise, matching the oracle.  The linear filter
      alone does not, even with a large number of samples.}
  \end{subfigure}
    ~ 
    \begin{subfigure}[t]{0.45\textwidth}
        \includegraphics[width=\textwidth]{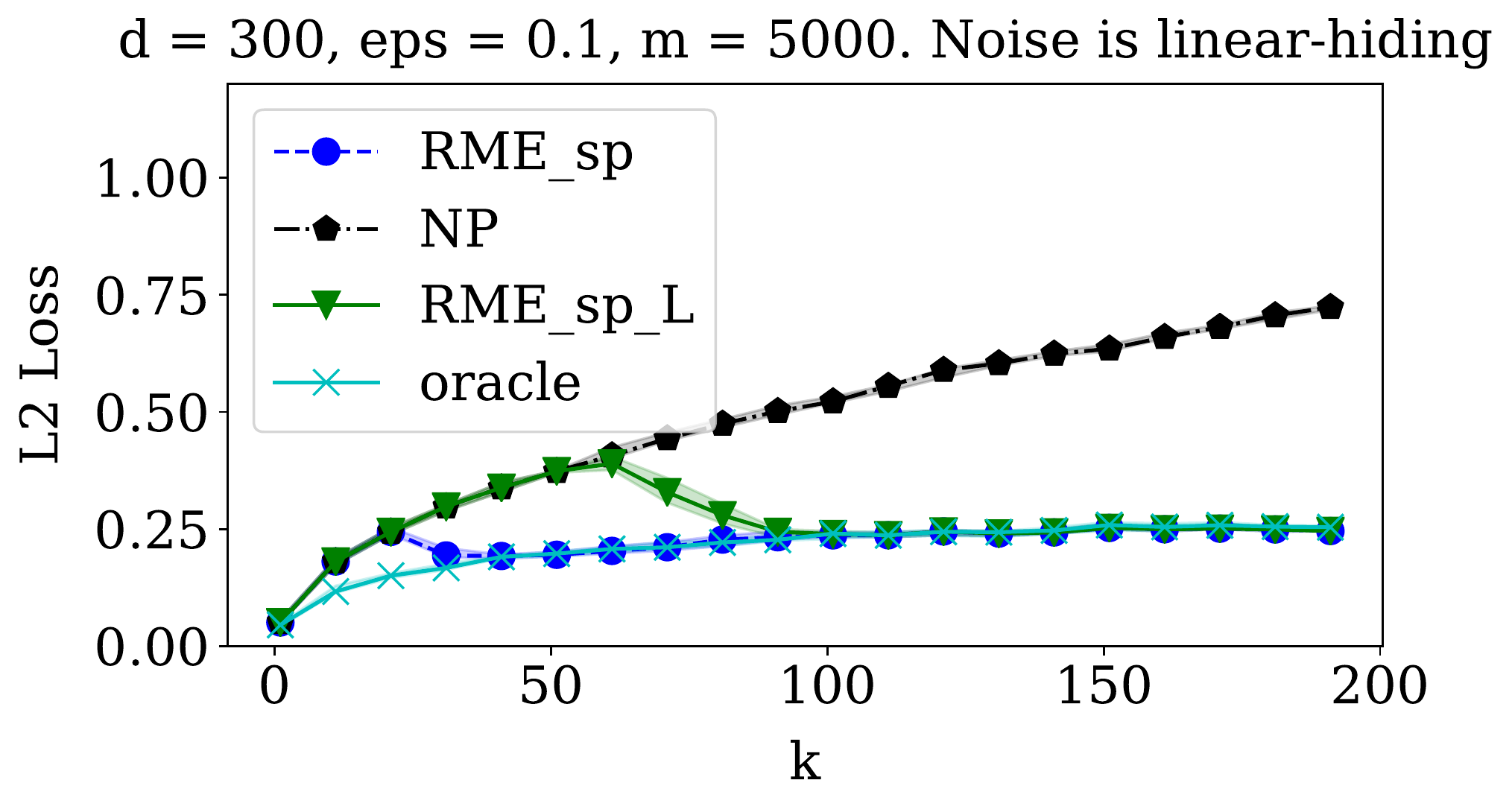}
        \caption{For $k \ll \sqrt{d}$, the linear filter alone does
          not filter out the noise, leading to an $\eps \sqrt{k}$
          dependence for \RMEspL.  Our algorithm \RMEsp{} nearly
          matches \oracle.}
    \end{subfigure}

    \caption{The linear-hiding noise model shows that the quadratic
      filter is necessary.}
  
  \label{fig:bimodal}

  \centering
  
  \begin{subfigure}[t]{0.45\textwidth}
 	\includegraphics[width=\textwidth]{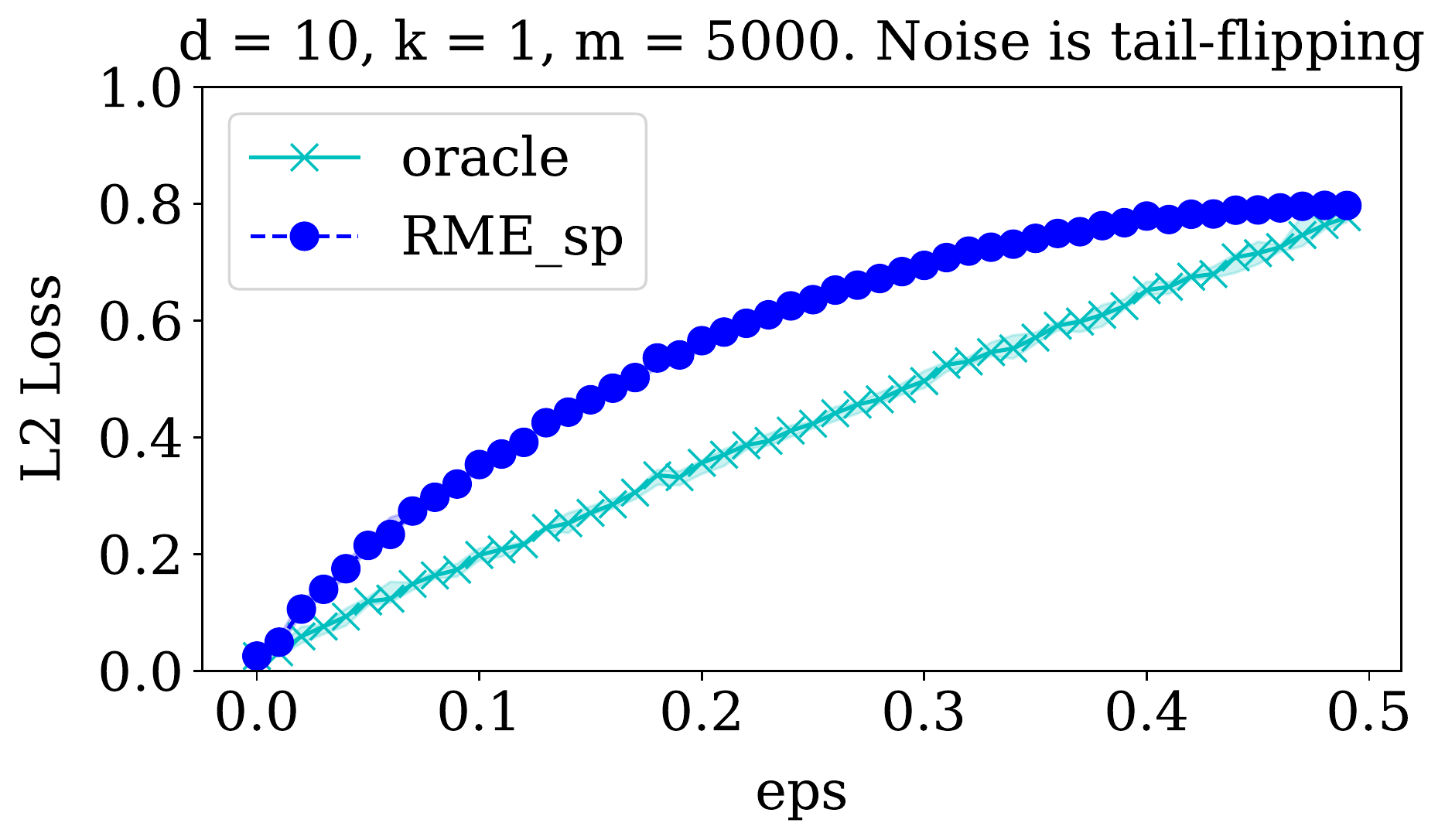}
  	\caption{This noise model gives $\Omega(\eps \sqrt{\log(1/\eps)})$ error
  		to the oracle, and \RMEsp{} is at most twice this.}
  \end{subfigure}
  ~ \begin{subfigure}[t]{0.45\textwidth}
    \includegraphics[width=\textwidth]{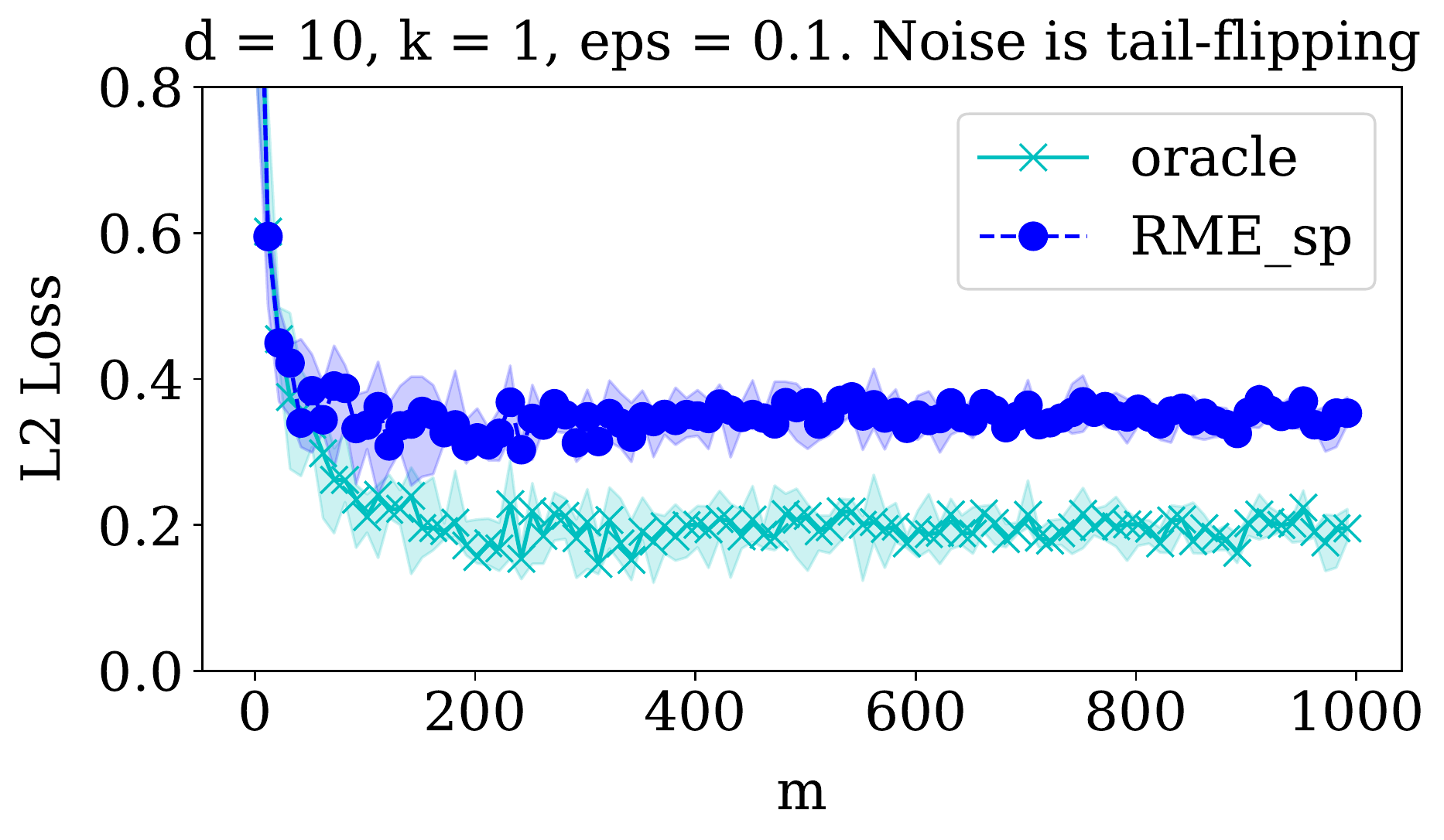}
    \caption{This gap persists regardless of $m$.}
  \end{subfigure}
    \caption{The flipping noise model demonstrates that the error can remain $\Omega(\eps \sqrt{ \log(1/\eps))}$.}
  \label{fig:flipping}
\end{figure}

\begin{figure}
	\centering
	
	\begin{subfigure}[t]{0.45\textwidth}
		\includegraphics[width=\textwidth]{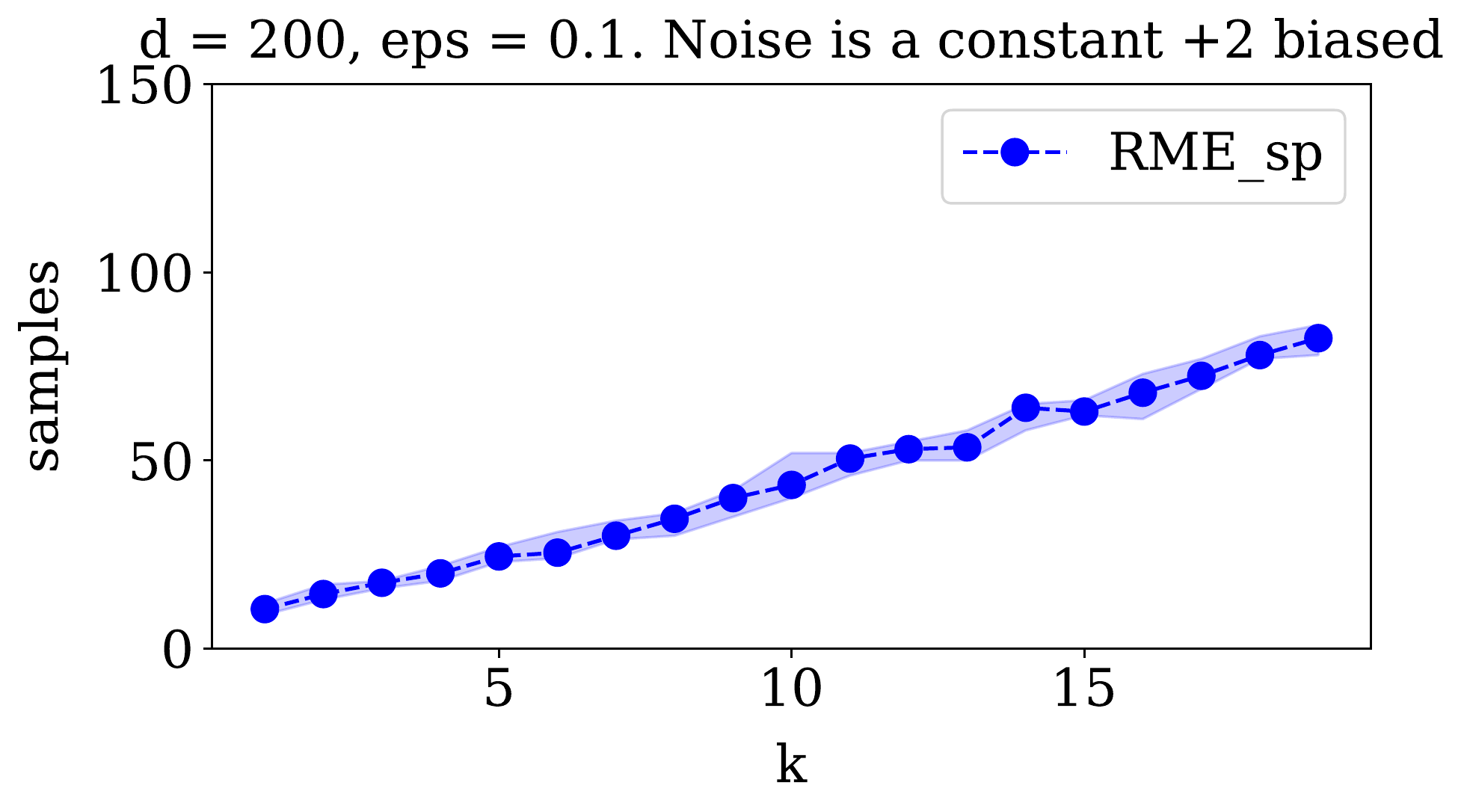}
		\caption{The constant noise model is easy to remove and does not take many samples.}
	\end{subfigure}
	~ \begin{subfigure}[t]{0.45\textwidth}
		\includegraphics[width=\textwidth]{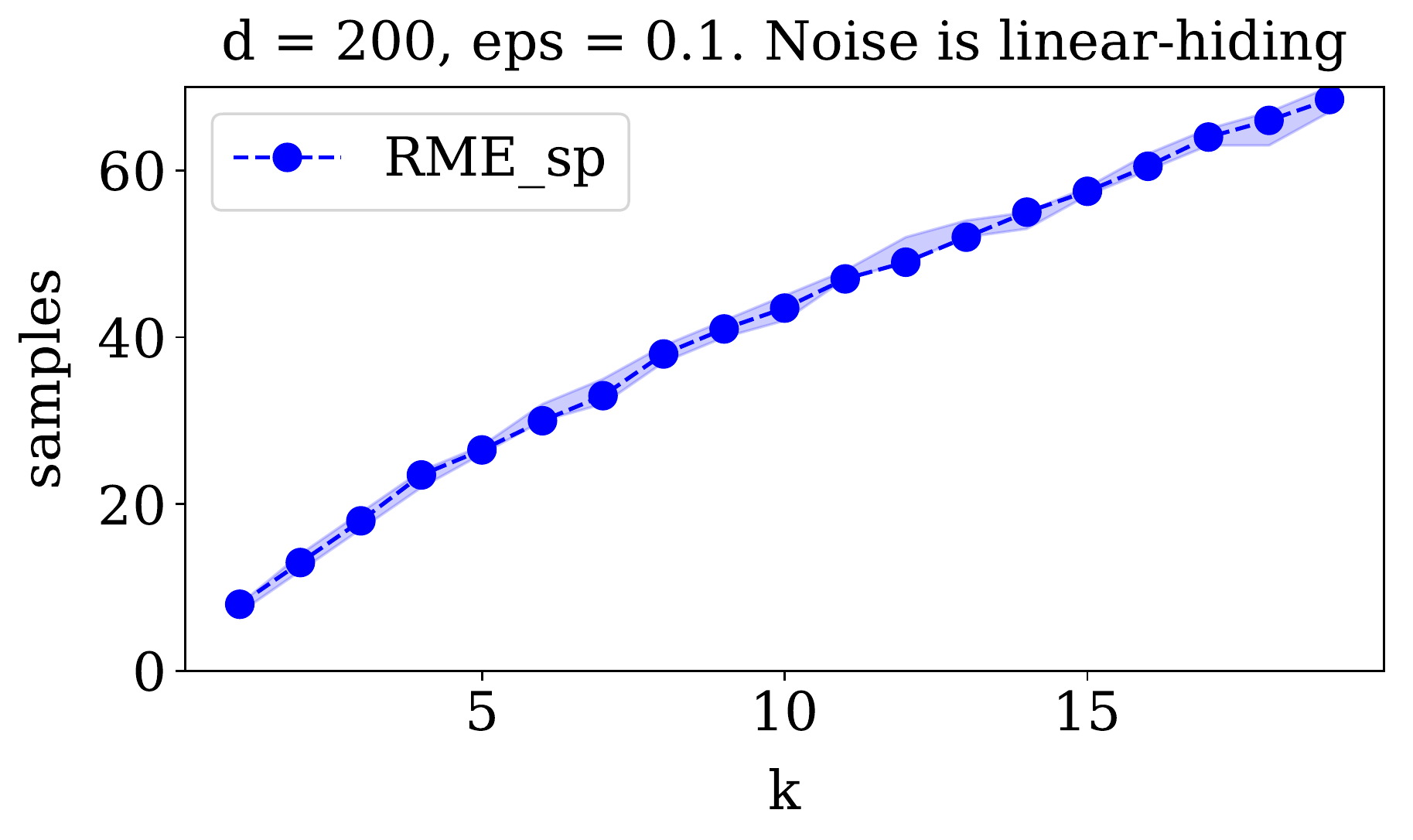}
		\caption{The linear-hiding noise model is harder and requires more samples to get the same guarantee.}
	\end{subfigure}
	\caption{Sample complexity required to do well---in this case, 70\% of errors being less than $1.2$---depends on the noise model.}
	\label{fig:sample_complexity}
\end{figure}

\begin{figure}
	\centering 
	  \begin{subfigure}[t]{0.45\textwidth}
		\includegraphics[width=\textwidth]{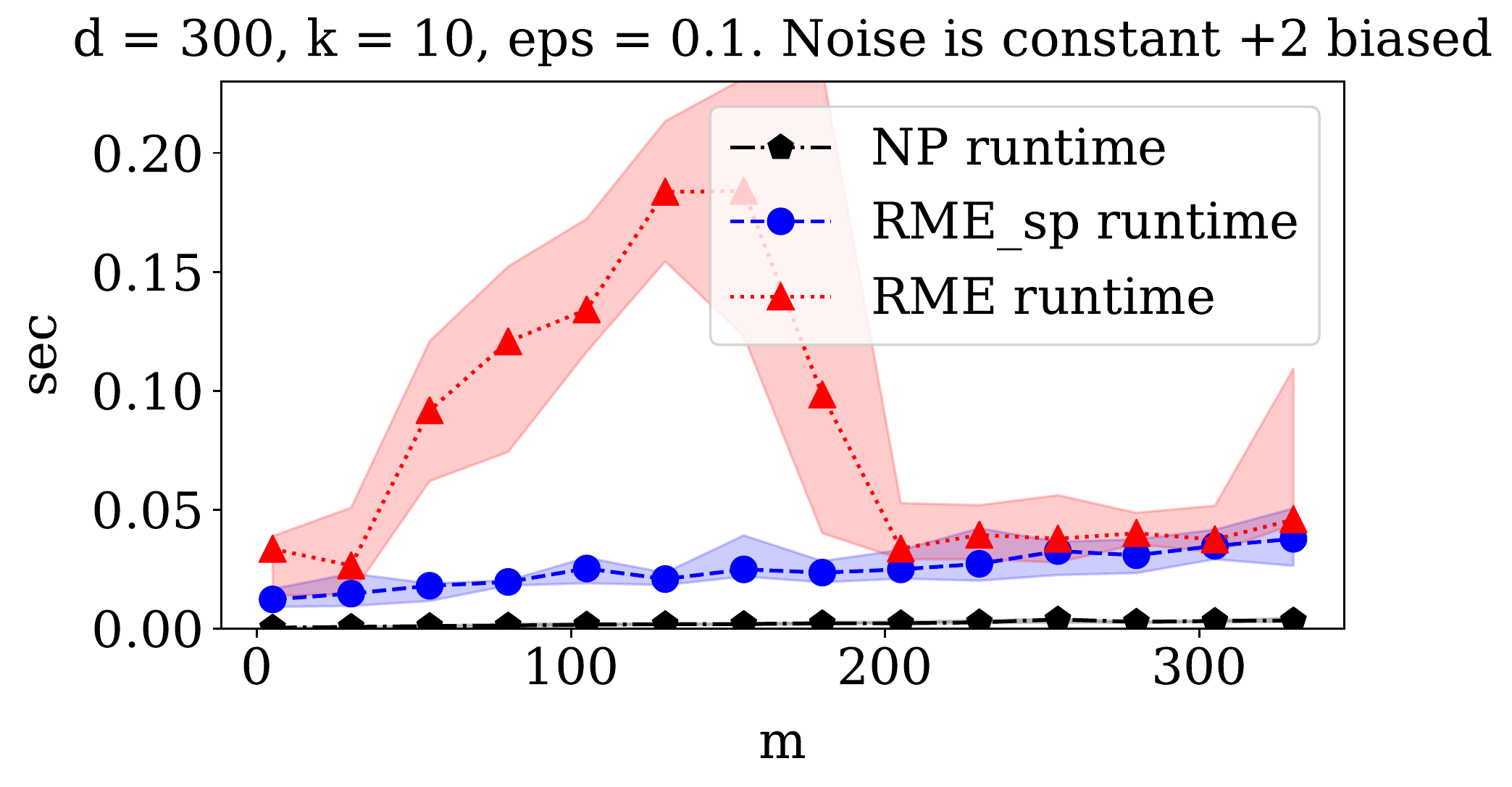}
		\caption{ Not only does \RMEnonsp{} not have small error for small sample complexity, interestingly it also takes longer to terminate.}  
	\end{subfigure}
	~ 
\begin{subfigure}[t]{0.45\textwidth}
 	\includegraphics[width=\textwidth]{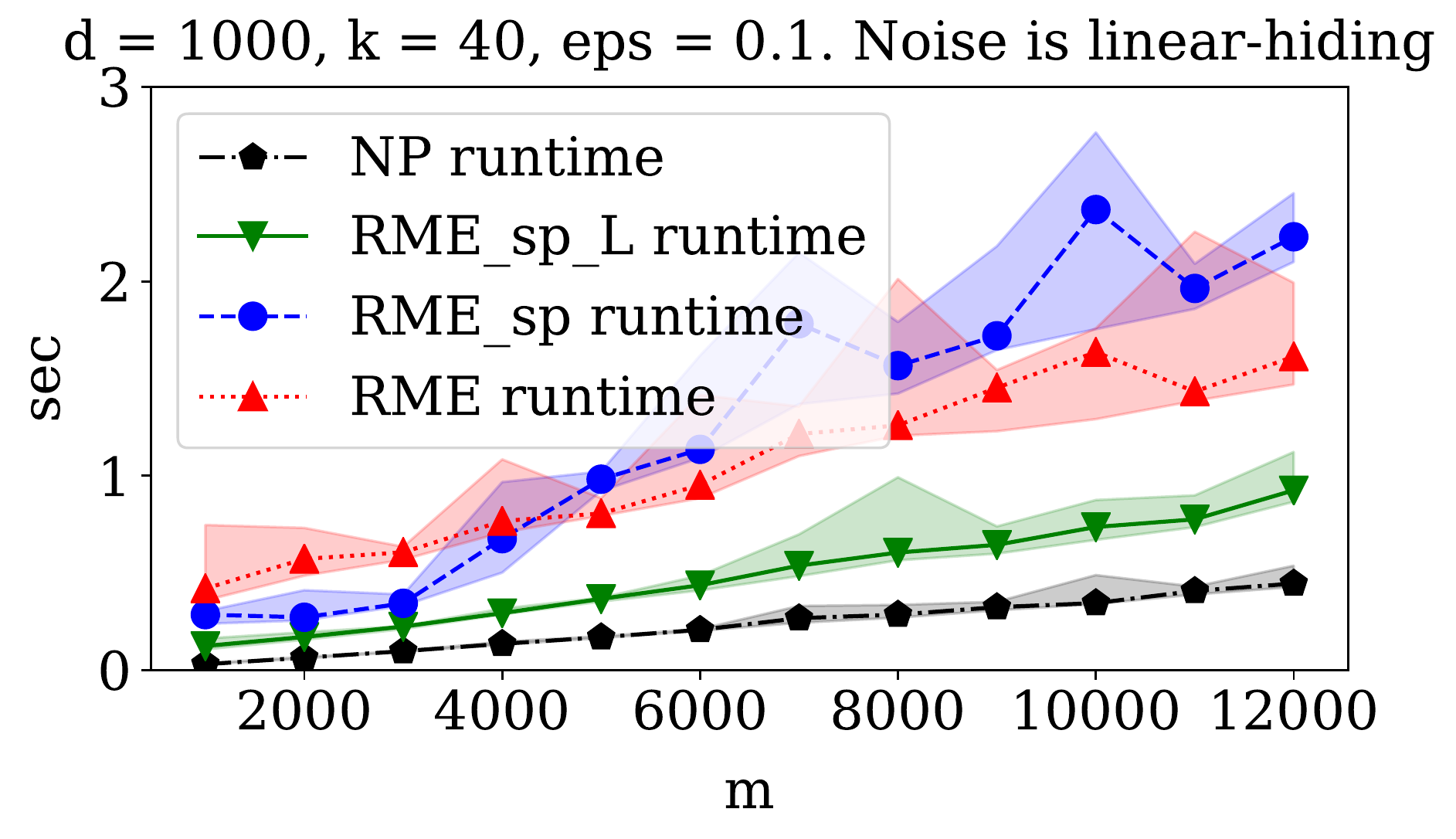}
	\caption{\RMEsp{} takes time close to \RMEspL{} until the quadratic filter begins to apply (as can be seen in Figure~\ref{fig:bimodal}), after which it takes much longer.}
\end{subfigure}
~ 

	\begin{subfigure}[t]{0.45\textwidth}
		\includegraphics[width=\textwidth]{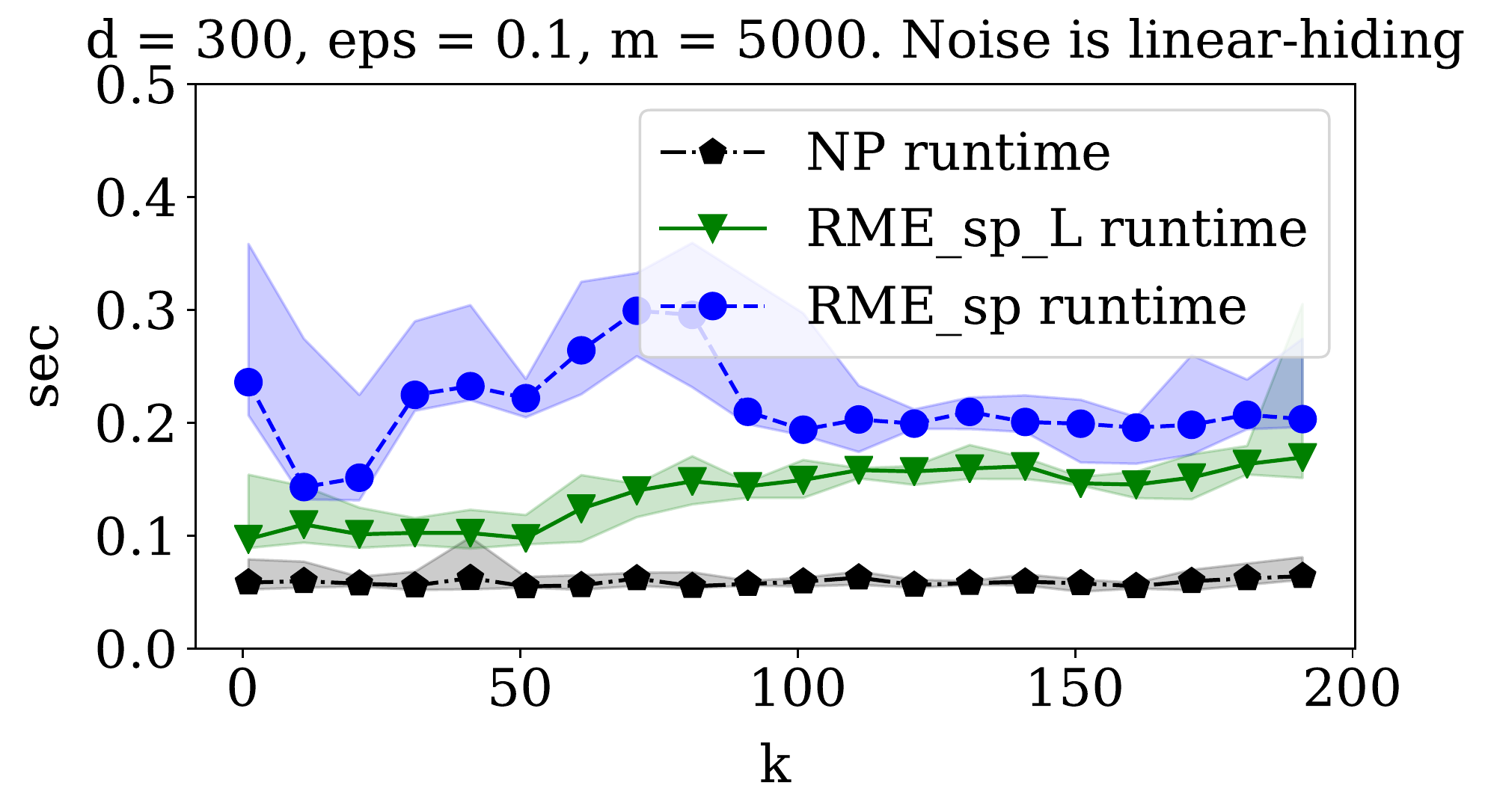}
		\caption{The runtimes for our sparse algorithms does not change very much as we increase $k$ for the linear-hiding noise.}
	\end{subfigure}
\begin{subfigure}[t]{0.45\textwidth}
	\includegraphics[width=\textwidth]{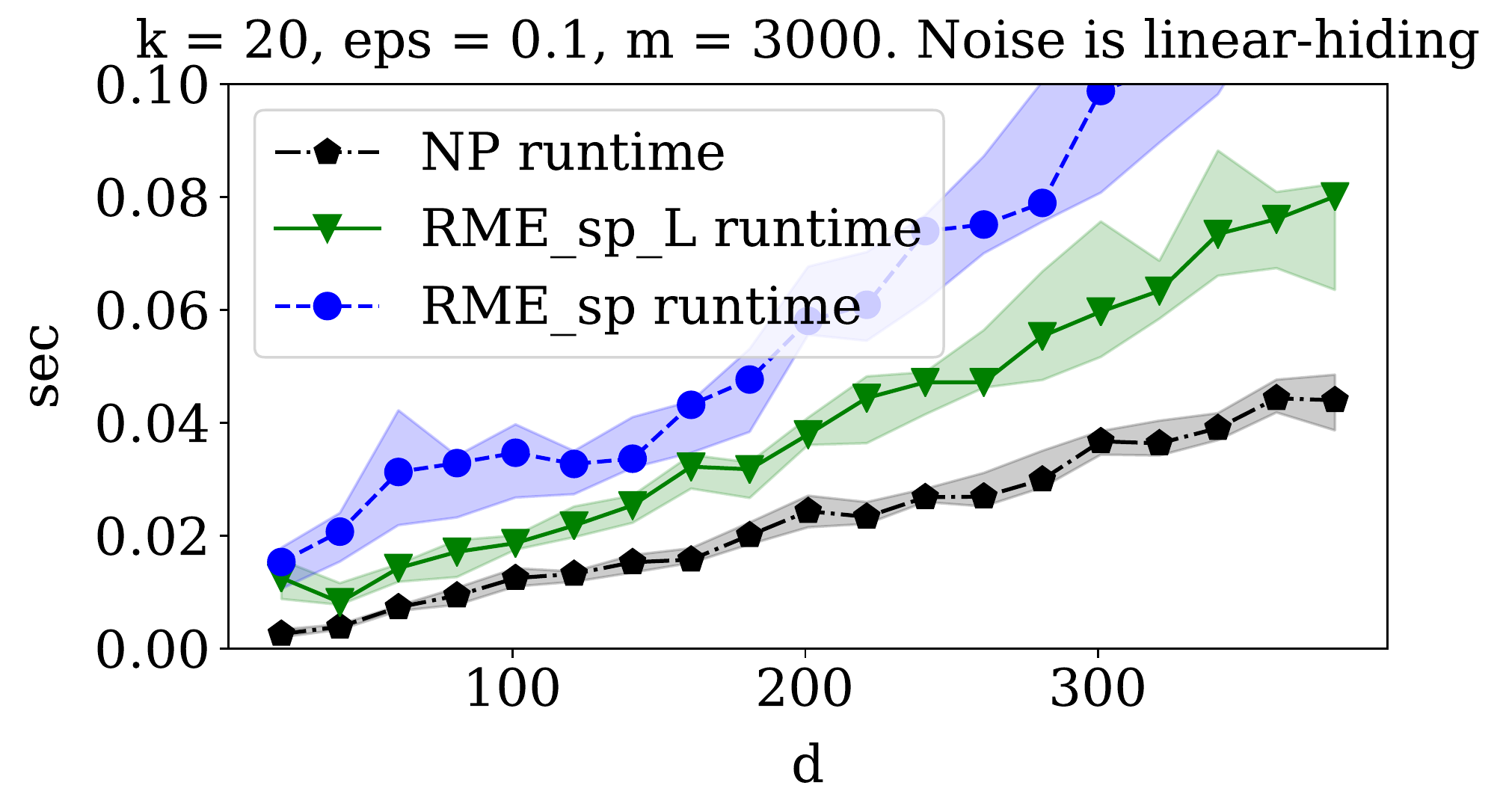}
	\caption{The runtimes for our sparse algorithms appears to increases with $d$ linearly the case of linear-hiding noise.}
	
\end{subfigure}

		\caption{Runtimes for robust mean estimation.}
\label{fig:runtimes}
\end{figure}

The performance of robust estimation algorithms depend heavily on the
noise model.  The ``hard'' noise distributions for one algorithm may
be easy for a different algorithm, if that one can identify and filter
out the outliers.  We therefore consider three different synthetic
data distributions: two that demonstrate the $\eps \sqrt{k}$
worst-case performance of other algorithms, and one that demonstrates
the $\eps \sqrt{\log(1/\eps)}$ performance of our full algorithm.

The algorithms we consider are \RMEsp{}, our algorithm; \RMEspL{}, a
version of our algorithm with only the linear filter and not the
quadratic one; \NPsp{}, the ``naive pruning'' algorithm that drops
samples with obviously-outlier coordinates, then outputs the empirical
mean; \oracle{}, which is told exactly which coordinates are inliers and
outputs their empirical mean; \RMEnonsp{}, which applies the non-sparse
robust mean estimation algorithm of~\cite{DKK+17}; and \RANSAC{}, which computes the mean of a randomly chosen set of points, half the size of the entire set. One mean is preferred to another if it has more points in a ball of radius $\sqrt{d + \sqrt{d}}$ around it. For algorithms that have non-sparse outputs, we sparsify to the largest $k$ coordinates before measuring the $\ell_2$ distance to the true mean.

Our distributions are:

\begin{itemize}
\item \textbf{Constant-bias noise.}  Noise that biases every
  coordinate consistently (e.g., if the outliers add $2$ to every
  coordinate, or set every coordinate to $\mu_i + 1$) is difficult for
  naive algorithms (such as coordinate-wise median, \NPsp{}, \RANSAC{})
  to deal with, but ideal for the linear filter.  In
  Figure~\ref{fig:const-bias} we consider the noise that adds $2$ to
  every coordinate.
\item \textbf{Linear-hiding noise.}  To demonstrate that the quadratic
  filter in our algorithm is necessary, we use the following data
  distribution.  The inliers are drawn from $\normalpdf(0, I)$.  The outliers
  are evenly split between two types: $\normalpdf(1_S, I)$ for some size-$k$
  set $S$, and $\normalpdf(0, 2I - I_S)$.  The diagonal of the empirical
  covariance does not reveal $S$, so our linear filter fails to prune
  anything, leading to $\eps \sqrt{k}$ error for \RMEspL{}; the quadratic filter
  successfully removes all the outliers.  This is shown in
  Figure~\ref{fig:bimodal}.
\item \textbf{Flipping noise.}  For both those types of noise, with
  sufficiently many samples our final algorithm will prune out
  essentially all the outliers; there also exist noise models where
  $\Omega(\eps\sqrt{\log(1/\eps)})$ noise will remain at all times.
  In Figure~\ref{fig:flipping} we demonstrate this for the noise model
  that picks a $k$-sparse direction $v$, and replaces the $\eps$
  fraction of points furthest in the $-v$ direction with points in the
  $+v$ direction.  In fact, for this noise even the \oracle{} method
  also has $\Omega(\eps \sqrt{\log (1/\eps)})$ error from the missing
  points, but our algorithm has twice the error from the unfilterable
  added points.
\end{itemize}


\paragraph{Discussion.}  Matching our theoretical results, with
sufficiently many samples the worst-case performance of \RMEsp{} seems
to be within a constant factor of the $O(\eps\sqrt{\log(1/\eps)})$
worst-case performance of \oracle{}.  This is not true for the naive
algorithms \NPsp, \RANSAC, or the simplification \RMEspL{} of our
algorithm, which all have an $\eps \sqrt{k}$ dependence. 
While our theoretical results show that $\Ot(k^2)$ samples suffice,
the empirical results given in Figure~\ref{fig:sample_complexity} are
consistent with $\Ot(k)$ being sufficient.

Our algorithm runs much faster than the ellipsoid based approach. For instance for $k = 10, d = 300, m = 50$ for the case of constant-biased noise our algorithm takes time $0.015$ seconds to finish. In comparison the very first iteration for the SDP-based solution takes $10$ seconds to solve with CVXOPT; the full ellipsoid-based algorithm, if implemented, would take many times that.

\subsection{Robust Sparse PCA}

\begin{figure}
	\centering
	
	\begin{subfigure}[t]{0.45\textwidth}
		\includegraphics[width=\textwidth]{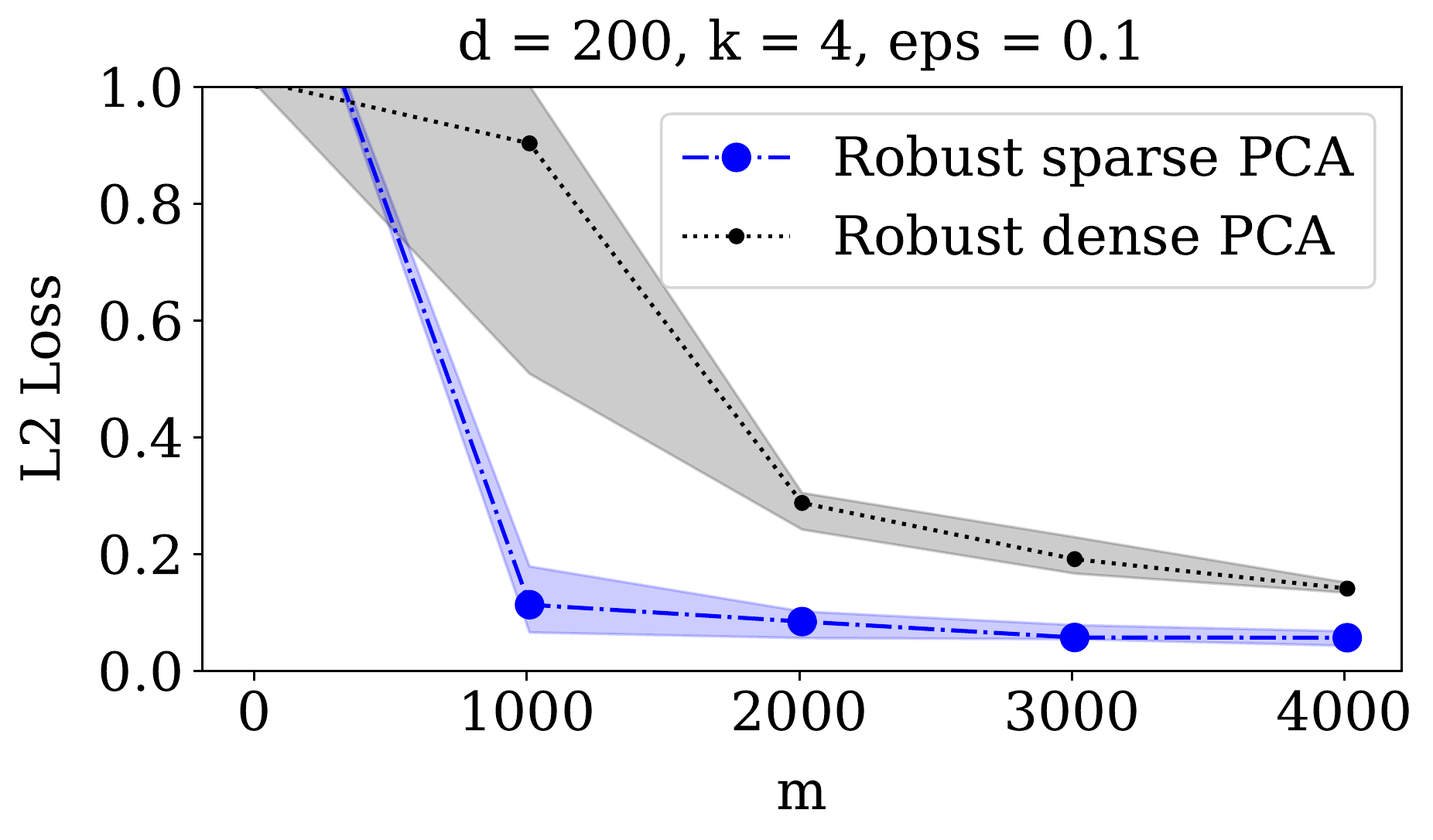}
		\caption{The natural dense algorithm \RDPCA{} requires more samples than the sparse algorithm to get error $< 0.1$}
	\end{subfigure}
~\begin{subfigure}[t]{0.45\textwidth}
	\includegraphics[width=\textwidth]{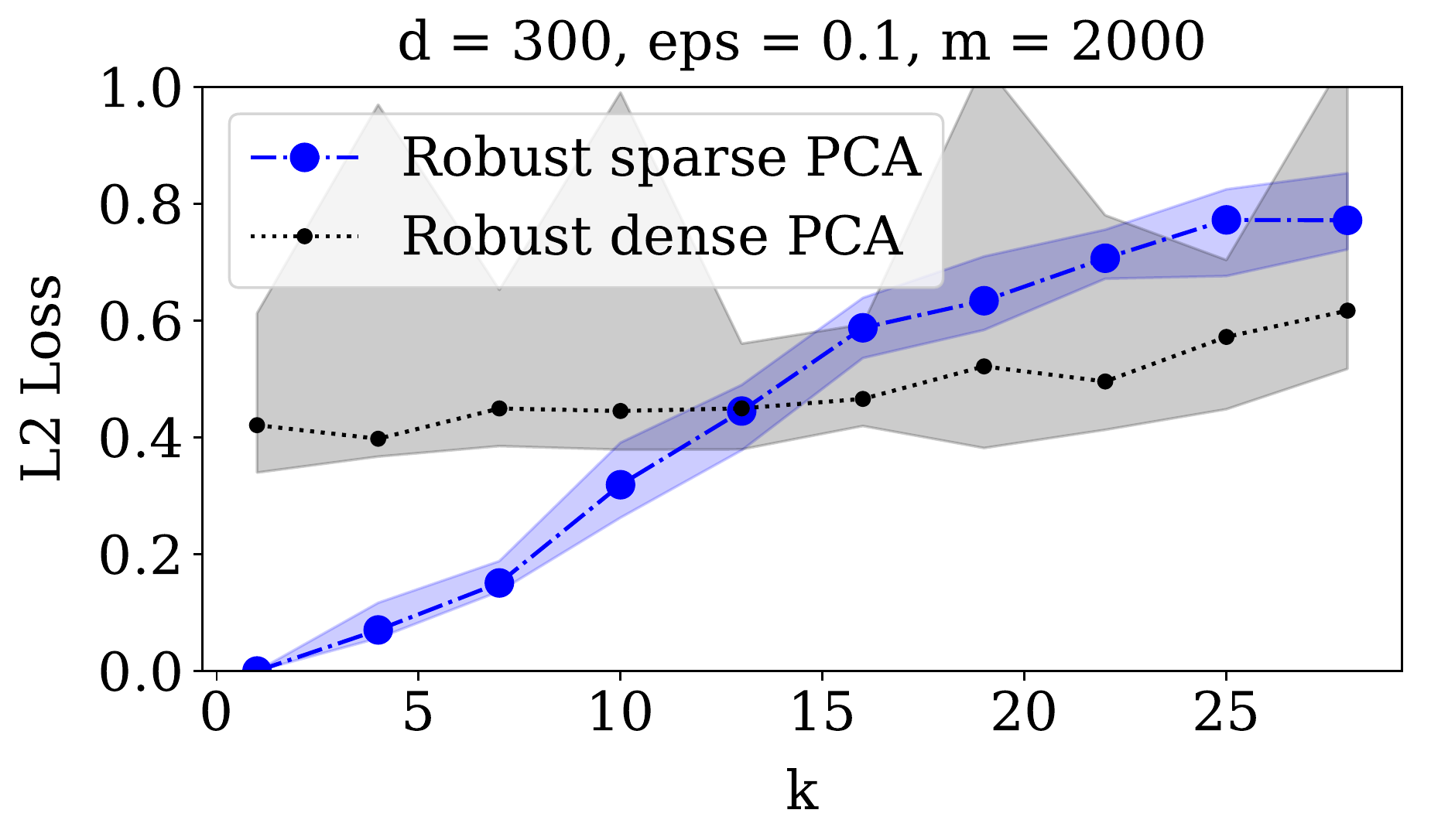}
	\caption{For a fixed $m$,  \RSPCA{} performs better than \RDPCA{} when $k < \sqrt{d}$ and then performs worse. until coming close to \RDPCA. Note that the variance of \RSPCA{} is smaller than that of \RDPCA{}.}
\end{subfigure}
	\caption{Sample complexity of \RSPCA{} is better than \RDPCA{} for smaller sparsity.}
	\label{fig:rspca}
\end{figure}

In Figure~\ref{fig:rspca} we compare our robust sparse PCA algorithm
\RSPCA{} to a dense algorithm \RDPCA{} for robust PCA. \RDPCA{} looks
at the empirical covariance matrix and then in the direction of
maximum variance robustly estimates standard deviation. The algorithm
then filters points using a modified version of the linear filter from
\cite{DKK+17} and hence requires a sample complexity of
$\tilde{O}(d)$.  For this algorithm, we only consider a single simple
noise model.  We draw outlier samples from $\normalpdf(0, I+uu^T)$
where $u$ has disjoint support from the true vector $v$.

The sparse algorithm seems to perform better than the dense algorithm
for $k$ up to roughly $\sqrt{d}$; this is better than what we can
prove, which is that it should be better up to at least $d^{1/4}$.


\section{Conclusion}

In this paper we have presented iterative filtering algorithms for two
natural robust sparse estimation tasks: sparse mean estimation and
sparse PCA.  In both cases, our algorithm achieves a near-optimal
$\Ot(\eps)$ error with a sample complexity primarily dependent on the
sparsity $k$, and only logarithmically on the ambient dimension $d$.
Our theoretical results are comparable to those of~\cite{BDLS17}, but
our algorithm only uses simple spectral techniques rather than the
ellipsoid algorithm.  This makes our algorithm quite feasible to
implement.  Our implementations perform essentially as expected: in
sparse settings they require significantly fewer samples than dense
robust estimation, and have accuracy avoiding the $\sqrt{k}$
dependence of other techniques like RANSAC.

\bibliographystyle{alpha}
\bibliography{allrefs}

%

\appendix
\section{ Proof of Lemma~\ref{lem:samples-good-mean}}

	Let $G$ be a set of $N=\todo{\widetilde{\Omega} \left(k^2 \log(d/\tau)/\eps^2 \right)}$
	i.i.d. samples drawn from $\normalpdf(\mu,I)$.  We will show that each of Conditions (i)-(iv)
	hold with probability at least $1-\tau/5$. The lemma then follows by a union bound.

	\paragraph{Proof of (i):}  
	To establish (i), let $\mu^G:= \E_{X \in_u G}[X]$ and note that the random variable $N \mu^G$ 
	is distributed as $\normalpdf(N \cdot \mu, N \cdot I)$. Hence, $\mu^G$ has independent coordinates 
	with $N \mu^G_i \sim \normalpdf(N \cdot \mu_i, N)$. By standard Gaussian tail bounds, we have that 
	$\pr \left[ \left| N (\mu^G_i - \mu_i) \right| \geq T \sqrt{N}\right] \leq 2 \cdot \exp (-T^2/2)$.
	Setting $T/\sqrt{N} = \eps/k$ gives that 
	$\pr \left[|\mu^G_i-\mu_i| \geq \eps/k \right] \leq 2 \cdot \exp (-N\eps^2/(2k^2)) \leq \tau/(10d).$
	By a union bound over all $i \in [d]$, it follows that 
	\[\pr \left[ \exists i \in [d] :  |\mu^G_i-\mu_i | \geq \eps/k \right]  \leq \tau/10 \;.\] 
	This completes the proof of the first part of (i).
	
	For the second part of (i), we will show that with probability at least $1-\tau/10$ we have that 
	for all $i, j \in [d]$, $|\E\left[(X_i-\mu_i)(X_j-\mu_j)\right] - \delta_{ij}| \leq \eps/k$.
	We will need the following simple technical fact:
	\begin{fact}[see, e.g.,~\cite{LM00}]
		Let $Y_i$ be iid standard univariate Gaussians and $a_i \geq 0$, $i \in [m]$. 
		If $Z = \sum_{i=1}^m a_i (Y_i^2 - 1)$, then  for any $x \geq 0$ the following hold:
		\begin{equation} \label{eqn:chi2}
		\pr \left[Z \geq 2 \|a\|_2 \sqrt{x} + 2\|a\|_{\infty}x \right] \leq \exp(-x) \;,
		\end{equation} 
		and 
		\begin{equation}
		\pr \left [Z \leq -2 \|a\|_2 \sqrt{x} \right] \leq \exp(-x) \;.
		\end{equation}
	\end{fact} 
	
	We start with the case that $i=j$. 
	Note that the random variable $N \cdot \E_{X \in_u G} \left[(X_i-\mu_i)^2 \right]$ follows a
	$\chi^2$-distribution with $N$ degrees of freedom, i.e., it is the sum of $N$
	independent squared standard Gaussians. 
	An application of Equation~\eqref{eqn:chi2} implies that for all $x \geq 0$ we have:
	\[\pr \left [ \left| N \cdot \E_{X \in_u G} \left[(X_i-\mu_i)^2 \right]-N \right| \geq 2 \sqrt{Nx} + 2x \right ] \leq \exp(-x).\] 
	Setting $x: = N\eps^2/(9k^2)$, we get that 
	\[ \pr \left[\left|\E_{X \in_u G} \left[(X_i-\mu_i)^2 \right]-1\right| \geq 2\eps/(3k) + 2\eps^2/(9k^2) \right] 
	\leq \exp\left(-N\eps^2/(9k^2)\right) \leq \tau/(10d^2).\] 
	
	We now analyze the case that $i \neq j$.
	Let $Y \sim \normalpdf(\mu,I)$. 
	Note that for $i \neq j$, $i, j \in [d]$, we have that 
	\[ (Y_i-\mu_i)(Y_j-\mu_j)=\left(\frac{(Y_i-\mu_i)}{2} + \frac{(Y_j-\mu_j)}{2}\right)^2 
	- \left(\frac{(Y_i-\mu_i)}{2} - \frac{(Y_j-\mu_j)}{2}\right)^2.\]
	Since $\frac{(Y_i-\mu_i)}{2} + \frac{(Y_j-\mu_j)}{2}$ and $\frac{(Y_i-\mu_i)}{2} - \frac{(Y_j-\mu_j)}{2}$ 
	are independent and distributed as $\normalpdf(0,1/2)$, for $i \neq j$,  the random variable 
	$N \cdot \E_{X \in_u G} \left[(X_i-\mu_i)(X_j-\mu_j)\right]$ 
	is distributed as the difference of a sum of $N$ independent squared zero-mean Gaussians 
	with variance $1/2$, and another such sum. 
	This random variable has expectation $0$ and once again, 
	by Equation~\eqref{eqn:chi2} applied with $a_i = 1/2$, it follows that 
	\[\pr \left[\left| N \cdot \E_{X \in_u G} \left[(X_i-\mu_i)(X_j - \mu_j) \right] \right| \geq 2 \sqrt{Nx} + x \right] \leq \exp(-x) \;.\]
	Setting $x: = N\eps^2/(9k^2)$ as above gives that
	\[ \pr \left[\left|\E_{X \in_u G} \left[(X_i-\mu_i)(X_j - \mu_j) \right] \right| \geq \eps/k \right] \leq \tau/(10d^2).\] 
	A union bound over all $i, j \in [d]$ implies that
	\[ \pr\left[ \exists i, j \in [d]: \left|\E_{X \in_u G} \left[ (X_i-\mu_i)(X_j - \mu_j)\right] -  \delta_{ij} \right| \geq \eps/k \right] \leq \tau/10.\] 
	This gives the second part of (i).
	By a union bound, Condition (i) holds with probability at least $1-\tau/5$.

	\paragraph{Proof of (ii):} 
	For $Y \sim \normalpdf(\mu,I)$,
	the standard Gaussian tail bound gives 
	$\pr \left[ \left|Y_i-\mu_i \right| \geq T\right] \leq 2 \exp\left(-T^2/2\right)$. 
	Setting $T=\sqrt{2 \ln\left(10Nd/\tau\right)}$ implies that
	$\pr \left[\left|Y_i-\mu_i\right| \geq T\right] \leq \tau/(10Nd)$. 
	By a union bound, the desired upper bound holds for all $i \in [d]$ and all $N$ samples
	with probability at least $1-\tau/10$.

	\paragraph{Proof of (iii):} 
	To establish (iii), we first prove that Conditions (iii)(a)-(c) hold for any fixed unit vector $v$ and threshold $T \geq 4$ 
	with sufficiently high probability, and then take a union bound over a net of $2k^2$-sparse unit vectors and thresholds. 
	
	To avoid clutter in the notation, we will denote $\delta \eqdef \frac{\eps^2}{\ln(k\ln(Nd/\tau))}$, so that 
	the second term in the RHS of Condition (iii)(c) is equal to $\delta/T^2$.
	
	We start by proving the following claim:
	\begin{claim} \label{clm:1Diii} 
		For any unit vector $v$ in $\R^d$ and threshold $T \geq 4$ 
		with probability at least $1-\exp\left(-\Omega\left(\frac{N\delta}{\log(1/\delta)}\right)\right)$, 
		we have that 
		(a) $|\E_{X \in_u G}[v \cdot (X-\mu)]| \leq O(\eps)$, (b) $|\E_{X \in_u G}[(v \cdot (X-\mu))^2]-1| \leq O(\eps)$, 
		and (c) $\pr_{X \in_u G}[|v \cdot (X - \mu)| \geq T] \leq (5/2) \cdot \erfc(T/\sqrt{2}) + \delta/(2T^2)$.
	\end{claim}
	\begin{proof}
		\todo{
			To prove (a), note that for each fixed unit vector $v \in \R^d$, 
			$N \E_{X \in_u G}[v \cdot (X-\mu)]$ is distributed as $\normalpdf(0, N)$.
			By standard Gaussian tail bounds, we have that 
			\[\pr \left[|\E_{X \in_u G}[v \cdot (X-\mu)]| \geq \eps \right] \leq 
			2 \cdot \exp (-N\eps^2/2) \ll \exp\left(-\Omega\left(N\delta/\log(1/\delta)\right)\right)\;,\]
			where the last inequality follows from the fact that $\delta \ll \eps^2$.
		}
		
		To prove (b), note that for each fixed unit vector $v \in \R^d$ the random variable 
		$N \cdot \E_{X \in_u G}[(v \cdot (X-\mu))^2]$ follows a $\chi^2$-distribution 
		with parameter $N$. By Equation~\eqref{eqn:chi2}, we get 
		\[ \pr \left[|N\cdot \E_{X \in_u G}[(v \cdot (X-\mu))^2]-N| \geq 2 \sqrt{Nx} + 2x \right] \leq \exp(-x) \;,\]
		for $x \geq 0$. Applying the above inequality for $x := N\eps^2/9$, we get  
		\[\pr \left[ \left|\E_{X \in_u G}[(v \cdot (X-\mu))^2\right]-1| \geq 2\eps/3+ (2/9)\eps^2 \right] \leq \exp\left(-N\eps^2/9 \right).\]
		
		To prove (c), we start by noting that, for any fixed unit vector $v$ and $Y \sim \normalpdf(\mu, I)$, 
		$v \cdot (Y - \mu)$ is a standard univariate Gaussian, and therefore 
		$\pr[|v \cdot (Y - \mu)| \geq T] = 2 \erfc(T/\sqrt{2})$. Let 
		\[Q(T) \eqdef (5/2) \erfc(T/\sqrt{2}) + \delta/(2T^2) \;.\] 
		Observe that $N \cdot \pr_{X \in_u G}[|v \cdot (X - \mu)| \geq T]$ 
		is a sum of $N$ independent Bernoulli random variables each with mean $2 \erfc(T/\sqrt{2})$. 
		An application of the Chernoff bound and the fact that 
		$Q(T) \geq (5/4) \left[ 2\erfc(T/\sqrt{2}) \right]$ gives that 
		$\pr_{X \in_u G}[|v \cdot (X - \mu)| \geq T] \geq Q(T)$ holds with probability at most $\exp\left(-\frac{N Q(T)}{60}\right)$.
		
		We choose $T'$ to satisfy $\erfc(T')= \delta^2/(4T'^4)$, which implies that  $T'=\Theta( \sqrt{\ln(1/\delta)})$. 
		We break the analysis into two cases: $T \leq T'$ or $T > T'$. 
		
		If $T \leq T'$, then $Q(T) \geq Q(T') \geq \delta/(2T'^2) = \Omega( \frac{\delta}{\log(1/\delta)})$ 
		and the above upper bound of $\exp\left(-\frac{N Q(T)}{60}\right)$ on the desired probability gives (c).
		
		If $T > T'$, we have that $\erfc(T/\sqrt{2}) \leq \delta^2/(4T^4)$. 
		In this case, we require a more precise version of the Chernoff bound, 
		which bounds from above the probability of the event 
		$\pr_{X \in_u G}[|v \cdot (X - \mu)| \geq T] \geq Q(T)$ 
		by $\exp\left(-N \cdot D_{KL}(Q(T)||2\erfc(T/\sqrt{2}))\right)$, 
		where $D_{KL}(p||q)$ denotes the KL-divergence between the Bernoulli random variables 
		with probabilities $p$ and $q$. 
		
		Let $p= \delta/(2T^2)$, $q=2 \erfc(T/\sqrt{2})$, and note that $q \leq p^2$ or $p/q \geq q^{-1/2}$. Then, since $T > 4$, we see $p < 1/32$ and we can now bound from below the KL-divergence by $\delta/10$, as follows:
		\begin{align*}
		D_{KL}(Q(T)||q) & \geq D_{KL}(p||q) = p\ln (p/q) + (1-p) \ln \left((1-p)/(1-q)\right) \\
		&\geq p\ln (p/q) - \ln(1-p) \geq p \left(\ln (p/q)-1- p \right) \\
		&\geq \delta/(2T^2) \cdot \left( \ln(1/q^{1/2}) - 1 - p)\right) \\
		& \geq \delta/(2T^2) \cdot \left(T^2/8 - T^2/16 \right) \geq \delta/16 \;,
		\end{align*}
		Thus, we have that $\pr_{X \in_u G}[|v \cdot (X - \mu)| \geq T] \geq Q(T)$ with probability at most 
		$\exp(-\Omega(N \delta))$ in this case. This completes the proof of (c).
		
		By a union bound, all events hold with probability at least 
		$1-\exp\left(-\Omega\left(\frac{N\delta}{\log(1/\delta)}\right)\right)$, 
		completing the proof of Claim~\ref{clm:1Diii}.
	\end{proof}

	We now define a cover over all $2k^2$-sparse vectors 
	as well as the possible values of $T$, and take a union bound over the product. 
	To this end, let \[R \eqdef \Theta\left(k\cdot \sqrt{\log\left (Nd/\tau\right)}\right)\] be such that by (ii) 
	we have $\|x-\mu\|_\infty \leq \frac{R}{\sqrt{2}k}$, for $x \in G$.

	For each set $U \subseteq [d]$ of coordinates of size $2k^2$, let $\mathcal{C}_U$ be an $\eps/R^2$-cover, in $\ell_2$-norm,
	of the set of unit vectors supported on $U$ (i.e., with all non-zero coordinates in $U$). 
	Such a cover exists with $|\mathcal{C}_U| \leq O\left(R^2/\eps\right)^{2k^2}$. 
	Let $\mathcal{C}$ be the union of $\mathcal{C}_U$ over all sets $U$ of coordinates of size $2k^2$. 
	Then we have that 
	\[|\mathcal{C}| \leq {d \choose 2k^2} \cdot O(R^2/\eps)^{2k^2} \leq O\left(dR^2/\eps\right)^{2k^2} \;.\]

	Let $\mathcal{T} := \{ \sqrt{i\eps} \mid i \in \Z_+, 0 \leq i \leq R^2/\eps^2\}$ be a net over thresholds $T$. 
	Note that $|\mathcal{C}| \cdot |\mathcal{T}| \leq O(dR^2/\eps)^{2k^2+2}$. 
	By a union bound, Claim~\ref{clm:1Diii} holds for all $v \in \mathcal{C}$ and $T \in \mathcal{T}$ 
	except with probability at most 
	\begin{align*}
	&O\left((dk^2/\eps)\log(Nd/\tau)\right)^{2k^2+2} \cdot \exp(\Omega(-N\delta/\log(1/\delta)) \\
	& =\exp\left(O(k^2 \log\left(dk \log(d/\tau)/\eps)\right) - \Omega\left(N\eps^2/\log^3(k/\eps\log(d/\tau))\right)\right) \leq \tau/10 \;,
	\end{align*}
	where we used the fact that $N=\todo{\widetilde{\Omega} \left(k^2 \log(d/\tau)/\eps^2 \right)}$.
	It remains to prove (iii) assuming this event holds.
	
	By definition, for any $k^2$-sparse unit vector $v \in \R^d$, there exists a $v' \in \mathcal{C}$ 
	such that $\|v'-v\|_2 \leq \eps/R^2$ and such that $v'-v$ is also $k^2$-sparse. 
	Thus, for any $x \in G$, we have
	\begin{align*} 
	|v \cdot (x-\mu)-v' \cdot (x-\mu)| &\leq \|v'-v\|_1 \|x-\mu\|_{\infty} \\
	&\leq \todo{\sqrt{2}} k \|v'-v\|_2 R/\sqrt{2}k \leq \eps/R \;.
	\end{align*}
	Therefore, for the mean we have that 
	$|\E_{X \in_u G}[v \cdot X]| \leq |\E_{X \in_u G}[v' \cdot X]| + \frac{\eps}{R} \leq O(\eps)$. 
	This gives Condition (iii)(a). 
	
	To establish Condition (iii)(b), we note that 
	for any $x \in G$, we have
	\begin{align*}
	|(v \cdot (x-\mu))^2-(v' \cdot (x-\mu))^2| &\leq 
	O\left( \left|v \cdot (x-\mu)-v' \cdot (x-\mu)\right| \left(|v \cdot (x-\mu)|+|v' \cdot (x-\mu)|\right) \right) \\
	&\leq O(\eps/R) \cdot O(k \cdot R/k) \\
	&\leq O(\eps) \;,
	\end{align*}
	where the second line uses the fact that $|v \cdot (x-\mu)| \leq \|v\|_1 \|x-\mu\|_{\infty} \leq k  \|x-\mu\|_{\infty} \leq R$.
	Therefore, we have that 
	$|\E_{X \in_u G}[(v \cdot (X-\mu))^2] -1| \leq |\E_{X \in_u G}[(v' \cdot (X-\mu))^2] -1| + O(\eps) = O(\eps)$. 
	This gives Condition (iii)(b).
	
	We now prove Condition (iii)(c).
	Consider the event $\{x \in G: |v \cdot (x-\mu)| \geq T\}$ for $T \geq \sqrt{2\ln(1/\eps)+2}$. 
	First note that this event is contained in the event $\{x \in G: |v' \cdot (x-\mu)| \geq T-\eps/R \}$. 
	Moreover, note that the event is empty, unless $T \leq \|v\|_1\|x-\mu\|_\infty \leq R$, in which case 
	$(T-\eps/R)^2 \geq T^2- \todo{2}\eps$. 
	Therefore, by the definition of $\mathcal{T}$, 
	there is a $T' \in \mathcal{T}$ with $T^2-\todo{2}\eps \leq T'^2 \leq (T-\eps/R)^2$.
	Then we have
	\begin{align*}
	\pr_{X \in_u G}\left[|v \cdot (X-\mu)| \geq T\right] &\leq 
	\pr_{X \in_u G}[|v' \cdot (X-\mu)| \geq T-\eps/\todo{R}] \\
	&\leq  \pr_{X \in_u G}[|v' \cdot (X-\mu)| \geq T'] \\
	& \leq 5\erfc(T')/2 + \delta/(2T'^2) \\
	& \leq 5 \erfc\left(\sqrt{T^2-\todo{2}\eps}\right)/2 + \delta/(2(\todo{T^2-2\eps})) \\
	& = (5/(2\sqrt{2 \pi})) \int_{\sqrt{T^2-\todo{2}\eps}}^\infty \exp(-x^2/2) dx +\delta/T^2\\
	& = (5/(2\sqrt{2 \pi})) \int_{T}^\infty \exp(-(y^2-\todo{2}\eps)/2) (y/\sqrt{y^2-\todo{2}\eps}) dy +\delta/T^2\\
	& = (5/(2\sqrt{2 \pi})) \int_{T}^\infty  \exp(\eps) \exp(-y^2/2) (1+O(\eps)) dy +\delta/T^2\\
	&\leq (5/(2\sqrt{2 \pi})) \int_{T}^\infty   (1+O(\eps)) \exp(-y^2/2) dy  +\delta/T^2\\
	& \leq 3 \erfc(T/\sqrt{2}) +\delta/T^2 \;,
	\end{align*}
	where the third line follows from Claim~\ref{clm:1Diii}(c) applied for $(v', T')$.
	This completes the proof of Condition (iii)(c).
	
	\paragraph{Proof of (iv):} 
	At a high-level, the proof is similar to that of Condition (iii) above.
	We start by proving that Conditions (iv)(a)-(b) hold for any fixed degree-$2$ polynomial 
	and threshold $T$ with sufficiently high probability, 
	and then take a union bound over a net of $k^2$-sparse $p(x)$ and $T$. 
	
	Note that a homogeneous degree-$2$ polynomial can be written as 
	$p(x)=(x-\mu)^T A (x-\mu)$, for a symmetric matrix $A$, 
	in which case we have $\E_{Y \sim \normalpdf(\mu, I)}[p(Y)]=\Tr(A)$ and 
	$\Var_{Y \sim \normalpdf(\mu, I)}[p(Y)]=\|A\|_F^{\todo{2}}$.
	
	We start by establishing the following claim:
	\begin{claim} \label{clm:1Div}
		Let $Y \sim \normalpdf(\mu, I)$. 
		Given a homogeneous degree-$2$ polynomial $p(x)$ with $\Var[p(Y)]=1$ and $T$ 
		with $4 \leq T \leq R\eqdef \Theta(k\cdot \sqrt{\log(Nd/\tau)})$, 
		we have that: (a) $\left|\E_{X \in_u G}[p(X)] - \E_{Y \sim \normalpdf(\mu, I)}[p(Y)]\right| \leq O(\eps)$, and
		(b) $\pr_{X \in_u G}[|p(X) - \E[p(Y)]| \geq T] \leq 2 \exp(-T/4)+\eps^2/(2 T \ln^2 T)$, except with probability 
		at most $\exp(-\Omega(N\eps^2/\ln^2 (R/\eps))$.
	\end{claim}

	\begin{proof}
		By diagonalizing $A$, we can write $p(Y)=c+\sum_{i=1}^{d} a_i Z_i^2$, 
		where the $Z_i$ are independent and distributed as $\normalpdf(0,1)$ and $c, a_i$ 
		are real coefficients with $\sum_i a_i^2= \|A\|_F^2=2$. 
		Note that $N \E[p(X)]$ is a sum of $Nd$ independent squared Gaussians, 
		each of which has variance at most $\|A\|_F^{\todo{2}}=1$ and the $\ell_2$-norm 
		of all their variances is $\sqrt{N}\|A\|_F=\sqrt{N}$. 
		Equation~\eqref{eqn:chi2} gives that $\pr[|N\E_{X \in_u G}[p(X)]-N\E_{Y \sim \normalpdf(\mu, I)}[p(Y)]| \geq 2 \sqrt{Nx} + 2 x] \leq \exp(-x)$, 
		for $x \geq 0$. Taking $x:=N \eps^2$, we obtain that 
		\[\pr\left[|\E_{X \in_u G}[p(X)]-\E_{Y \sim \normalpdf(\mu, I)}[p(Y)]| \geq 2\eps + 2 \eps^2 \right] \leq \exp(-N\eps^2).\]
		This shows (a).
		
		We proceed to prove (b).
		By Equation~\eqref{eqn:chi2} applied for a single sample, we have that 
		$\pr_{Y \sim \normalpdf(\mu, I)}[|p(Y) - \E_{Y \sim \normalpdf(\mu, I)}[p(Y)]| \geq 2 \sqrt{x} + 2 x] \leq \exp(-x)$ for $x \geq 0$. 
		Taking $x:=T$ for $T \geq 4$, we have $2 \sqrt{T} \leq T$, and so 
		\[\pr_{Y \sim \normalpdf(\mu, I)}[|p(Y) - \E_{Y \sim \normalpdf(\mu, I)}[p(Y)]| \geq T] \leq \exp(-T/4).\] 
		Note that $N\pr_{X \in_u G}[|p(X) - \E_{Y \sim \normalpdf(\mu, I)}[p(Y)]| \geq T]$ is a sum of $N$ independent Bernoulli random variables
		each with expectation at most $\exp(-T/4)$. 
		Let \[Q(T) \eqdef 2 \exp(-T/4)+\eps^2/(2 T \ln^2 T) \;.\] 
		Since $Q(T) \geq 2 \pr[|p(Y) - \E[p(Y)]| \geq T]$, 
		by the multiplicative Chernoff bound we have that 
		$\pr_{X \in_u G}[|p(X) - \E_{Y \sim \normalpdf(\mu, I)}[p(Y)]| \geq T] \leq Q(T)$, except with probability at most $\exp(-NQ(T)/6)$. 
		
		Let $T'$ be such that $\exp(-T'/6) = \eps^2/(2 T' \ln^2(T'))$. 
		Note that $T'=\Theta(\log(1/\eps))$. 
		For $T \leq T'$, we have that $Q(T) \geq \eps^2/(T' \ln^2{T'})$, 
		and so $\exp(-NQ(T)/6) \geq \exp(\Omega(-N\eps/\log(1/\eps)(\log \log(1/\eps))^2))$.
		
		For $T \geq T'$, note that $\eps^2/(2T \ln^2(T)) \geq \exp(-T/6)$. 
		Again we need to use a more explicit version of the Chernoff bound, 
		which gives that $\pr_{X \in_u G}[|v \cdot (X - \mu)| \geq T] \geq Q(T)$ 
		with probability at most $\exp(-N D_{KL}(Q(T)||\exp(-T/4)))$. 
		
		When $T' \leq T \leq R$, $p=\eps^2/(2T\ln^2 T)$, and $q=2\exp(-T/4)$, we obtain
		\begin{align*}
		D_{KL}(Q(T)||q) &\geq D_{KL}(p||q) = p\ln(p/q) + (1-p) \ln((1-p)/(1-q)) \\
		&\geq p\ln(p/q) - \ln(1-p) \\
		&\geq p (\ln(p/q) - 1 -p) \\
		&= (\eps^2/(2T \ln^2 T)) (\ln(p/\exp(-T/4))-1-p) \\
		&\geq (\eps^2/(2T \ln^2 T)) (\ln(\exp(T/6))-1-p) \\
		&\geq (\eps^2/(2T \ln^2 T)) \cdot (T/7)\\
		&\geq \eps^2/(14 \ln^2 T) \geq \eps^2/(14\ln^2 R) \;,
		\end{align*}
		Where we used the fact that $4 \leq T \leq R$, and the fact that if $p < 1/8$, then $\ln(1-p) < p (1-p)$. Thus, it follows that $\pr_{X \in_u G}[|v \cdot (X - \mu)| \geq T] \geq Q(T)$ with probability at most $\exp(-\Omega(N\eps^2/\ln^2 R))$ in this case.
		In either case, by a union bound, the claim holds except with probability $\exp(-\Omega(N\eps^2/\ln^2 (R/\eps)))$.
		This completes the proof of (b) and of Claim~\ref{clm:1Div}.
	\end{proof}
	
	It remains to construct a cover of $k^2$-sparse homogeneous degree-$2$ polynomials 
	which have at most $k^2$ terms and $\Var[p(Y)]=1$. Let $U$ be the set of $k^2$ monomials $x_ix_j$, for $1 \leq i,j \leq d$. 
	We construct a cover $\mathcal{C}_U$ of polynomials with terms only in the monomials in $U$ as follows:
	We take a cover of unit vectors in $\R^{k^2}$ to within $\ell_2$-norm $\eps/R^2$ 
	and use the coordinates of each vector as the coefficients of the corresponding monomial. 
	Thus, we can take $|\mathcal{C}_U| = 2^{O(k^2)}$. 
	Then we let $\mathcal{C}$ be the union of $\mathcal{C}_U$ for all sets of $k^2$ monomials $U$. 
	We therefore have that $|\mathcal{C}| \leq {d \choose k^2} \cdot O(R^2/\eps)^{k^2} \leq O(dR^2/\eps)^{k^2}$. 
	
	Let $\mathcal{T} = \{i\eps: i \in \Z_+, 0 \leq i \leq R^2/\eps^2  \}$. 
	Thus, $|\mathcal{C}| \cdot |\mathcal{T}| \leq O(dR^2/\eps)^{k^2+1}$. 
	By a union bound, Claim~\ref{clm:1Div} holds for all $p \in \mathcal{C}$ and $T \in \mathcal{T}$, 
	except with probability at most 
	\begin{align*}
	&O(dk^2\log(Nd/\tau)/\eps)^{k^2+1} \cdot \exp(-\Omega(N\eps^2/\ln^2 (R/\eps)))\\
	& =\exp\left(O(k^2 \log(dk \log(N/\tau)/\eps)) - \Omega(N\eps^2/\ln^2 (k\log(Nd/\tau)/\eps))\right) \leq \tau/10 \;,
	\end{align*}
	where we used the fact that $N=\todo{\widetilde{\Omega} \left(k^2 \log(d/\tau)/\eps^2 \right)}$.
	It remains to prove (iv) assuming this event holds.
	
	Consider any homogeneous degree-$2$ polynomial $p(x)$ 
	with at most $k^2$ terms and $\Var[p(Y)]=1$. 
	By construction of the cover, there is a polynomial $p'(x) \in \mathcal{C}$ such that the total number of monomials 
	appearing in either $p(x)$ or $p'(x)$ is at most $k^2$, 
	and if we write $p(x)=(x-\mu)^T A (x-\mu)$ and $p'(x)=(x-\mu)^T A' (x-\mu)$ for symmetric matrices $A, A'$, 
	then $\|A-A'\|_F \leq \eps/R^2$. Let $U'$ be the set of coordinates appearing in either $p(x)$ and $p'(x)$
	and note that $|U'| \leq 2k^2$.
	For $x \in G$, we have
	\begin{align*}
	|p(x)-p'(x)| & = |(x-\mu)^T (A-A') (x-\mu)| \\
	&=|(x-\mu)_{U'}^T (A-A') (x-\mu)_{U'}| \\
	& \leq \|(x-\mu)_{U'}\|_{\todo{\infty}}^2 \|A-A'\|_2 \\
	& \leq R^2 \cdot \eps/R^2 \leq \eps \;.
	\end{align*}
	Therefore, we have that 
	$|\E[p(X)] - \E[p(Y)]| \leq |\E[p'(X)]-\E[p'(Y)]|+2\eps \leq O(\eps)$, 
	since Claim~\ref{clm:1Div} holds for $p'(x)$.
	We have thus established Condition (iv)(a).
	
	To show Condition (iv)(b), consider the event $\{x \in G: |p(x)| \geq T\}$ for $T > 5$. 
	Let $T' \in \mathcal{T}$ be such that $T-2\eps \leq T' \leq T-\eps$. Then, 
	$|p(x)| \geq T$ implies that $|p'(x)| \geq T'$, and therefore
	\begin{align*}
	\pr_{X \in_u G}[|p(X)| \geq T] &\leq \pr_{X \in_u G}[|p'(X)| \geq T'] \\
	& \leq 2 \exp(-T'/3)+\eps^2/(2 T' \ln(T')^2) \\
	& \leq 2 \exp(-(T-2\eps)/3) + \eps^2/(2 (T-2\eps) \ln^2 (T-2\eps)) \\
	& \leq 3 \exp(-T/4) + \eps^2/(T \ln^2 T) \;,
	\end{align*}
	where the second line follows from Claim~\ref{clm:1Div}(b) for $(p', T')$.
	This completes the proof of Condition (iv)(b).
	
	The proof of Lemma~\ref{lem:samples-good-mean} is now complete.

\section{Proof of Lemma~\ref{lem:samples-good-pca}}

	Condition~\ref{cond:coord_dev} follows from standard gaussian concentration bounds.
	To see that Condition~\ref{cond:mean} holds, we prove entrywise closeness of the matrices involved. We will use the following standard concentration inequality
	\begin{lemma}\label{lem:gauss_poly_consc}
		For any degree $q$, $n$-variate polynomial $f$
		\[ \pr_{A \sim \normalpdf(0, \Sigma)} \left[ \left|f(A) - \mathbb{E}[f(A)] \right| > \tau \right] \lesssim e^{-\left(\frac{\tau^2}{R \cdot \var_{A \sim \normalpdf(0, \Sigma)}[f(A)]}\right)^{1/q}} \] where $R$ is some universal constant.
	\end{lemma}
	
	Entries of $xx^T-(I+\rho vv^T)$ are degree 2 polynomials of Gaussians, and thus so is their mean over $G$. Hence, Lemma~\ref{lem:gauss_poly_consc} implies that for any $(i,j)\in [d]\times [d]$ that
	\[\pr\left[\left| \E_{G}[x_ix_j] - (\delta_{i,j} + \rho v_iv_j)\right| > \eps/k \right] \lesssim \exp\left( -(N\eps^2/Rk^2)^{1/2} \right).\] Taking a union bound over $i,j$ shows that with high probability $\E_G[xx^T]$ has each entry within $\eps/k$ of that of $\rho vv^T+I$, and this immediately implies Condition~\ref{cond:mean}.
	
	Condition~\ref{cond:var} holds via a similar argument. Observe that it is sufficient to consider the case  $\|w\|_2=1$ and sample enough points to satisfy
	\[ |\E_G[(x_ix_j - \delta_{i,j}-\rho v_iv_j)(x_kx_l - \delta_{k,l}-\rho v_kv_l)] - \E_{\mathcal{N}(0, I+\rho vv^T)}[(x_ix_j - \delta_{i,j}-\rho v_iv_j)(x_kx_l - \delta_{k,l}-\rho v_kv_l)] | \leq \frac{\eps}{k^2}. \]
	Then the spectral norm of the covariance matrix of $\gamma(x)$ for any $Q \times Q$ submatrix will also be bounded by $\eps$. Note that this is just the probability that a degree-$4$ polynomial in Gaussian inputs deviates too much from its mean, and thus by Lemma~\ref{lem:gauss_poly_consc} the probability that the above fails to hold for any $(i,j,k,l)$ is at most 
	
	\[
	\exp\left( -(N\eps^2/Rk^4)^{1/4} \right).
	\]
	Taking a union bound over $(i,j,k,l)$ yields our result.
	
	Finally, for Condition \ref{cond:concentration}, we note that (perhaps changing the constant $C$), it suffices to prove it for all $\binom{d^2}{k}$ possible $Q$'s and for all $w$ in a cover of the unit ball of $\R^{k^2}$ (which will have size $2^{O(k^2)}$ and for $T$ powers of $2$ less than or equal to $k\log(dN)$ (since by Condition \ref{cond:coord_dev} $|\vvec(xx^T)_Q| = O(k\log(dN))$ for all $x\in G$). Once we have fixed $Q,w$ and $T$, $\gamma_Q(x) \cdot  w -  \rho \vvec(vv^T)_Q \cdot w$ is a mean $0$, variance $O(1)$, degree-$2$ polynomial so by Lemma \ref{lem:gauss_poly_consc}, the probability that it is more than $CT$ is at most $e^{-2T}$. Then the probability that at least $\eps N/(T^2 \log^2(T))$ of our $x$'s have this property is at most
	\begin{align*}
	\binom{N}{\eps N/(T^2\log^2(T))}\exp(-2T(\eps N)/(T^2 \log^2(T))) & \leq \left(\frac{Ne^{-2T}}{e \eps N/(T^2\log^2(T))}\right)^{(\eps N)/(T^2 \log^2(T))} \\ & \leq \exp(-\Omega(T \eps N/ (T^2 \log^2(T))))\\ & \leq \exp(-\Omega(\eps N/ T \log^2(T)))\\
	& \leq \exp(-\Omega(k^3 \log(d/\eps))).
	\end{align*}
	Taking a union bound over $Q,w,T$ completes the proof.

\end{document}